\documentclass[12pt]{article}

\usepackage{amsmath}
\usepackage{amsthm}
\usepackage{graphicx}
\usepackage{enumerate}
\usepackage{natbib}
\usepackage{url}
\usepackage{newtxtext}
\usepackage[subscriptcorrection]{newtxmath}
\usepackage[below]{placeins}
\usepackage{algorithm}
\usepackage{algorithmicx}       
\usepackage{algpseudocode}
\usepackage{bm}
\usepackage{xcolor}
\usepackage{subcaption}

\newtheorem{theorem}{Theorem}
\newtheorem{proposition}{Proposition}
\newtheorem{lemma}{Lemma}

\theoremstyle{remark}
\newtheorem{remark}{Remark}
\newtheorem{condition}{Condition}

\graphicspath{{./art/}}

\newcommand{\bx}{\bm x}

\newcommand{\bbeta}{\bm \beta}

\newcommand{\btheta}{\bm \theta}

\newcommand{\EE}{\mathbb E}

\newcommand{\MM}{\mathcal M}

\DeclareMathOperator{\var}{var}
\DeclareMathOperator{\cov}{cov}
\DeclareMathOperator{\logit}{logit}

\newcommand{\blind}{1}

\begin{document}

\def\spacingset#1{\renewcommand{\baselinestretch}{#1}\small\normalsize}
\spacingset{1}

\if1\blind
{
  \title{\bf Design-Based Cross-Validation for Comparing Small Area Estimators}
  \author{Qianyu Dong\thanks{
    We are grateful to Jon Wakefield, Sho Kawano, Paul Parker, 
Ameer Dharamshi, Geir-Arne Fuglstad, and members of the Space Time Analysis Bayes
(STAB) working group at the University of Washington for discussion and feedback on
the paper. We also acknowledge the DHS for access and
use of the data. The authors are supported by the National Institutes of Health [R01HD112421], and in part by the Gates Foundation. The findings and conclusions contained within are those of the authors and do not necessarily reflect positions or policies of the Gates Foundation.}\hspace{.2cm}\\
    Department of Statistics, University of California, Santa Cruz\\
    and \\
    Zehang Richard Li \\
    Department of Statistics, University of California, Santa Cruz}
  \maketitle
} \fi

\if0\blind
{
  \bigskip
  \bigskip
  \bigskip
  \begin{center}
    {\LARGE\bf Design-Based Cross-Validation for Comparing Small Area Estimators}

\end{center}
  \medskip
} \fi

\bigskip
\begin{abstract}
Subnational monitoring of public health and development indicators often relies on household surveys where data are sparse at the desired spatial resolution. Small area estimation (SAE) methods address this challenge by borrowing strength across areas and incorporating auxiliary information. However, comparing these estimators remains difficult in the absence of ground truth. We propose a design-based cross-validation framework for evaluating small area estimators that accommodates complex survey designs. Our approach enables model-agnostic comparisons between area-level and unit-level SAE models. We derive a decomposition of the conditional mean squared error that yields a consistent cross-validation score, show that finite-sample comparisons carry an unidentifiable bias that can be bounded, and use this bound as a principled threshold for ranking models.  We further show that leave-one-area-out cross-validation, a popular alternative, targets extrapolation rather than smoothing error and can reverse the correct ranking.  We evaluate the framework through extensive design-based simulations. We apply the framework to compare subnational female literacy estimators in Zambia using the 2024 Demographic and Health Survey. The framework applies broadly across prevalence mapping and other SAE problems and is applicable to any small area estimator irrespective of the underlying model class.
\end{abstract}

\noindent%

\vfill

\newpage
\spacingset{1.9} 

\section{Introduction}\label{sec:intro}
Small area estimation (SAE) refers to the process of estimating quantities of interest for specific geographic areas or subpopulations using survey data.  When within-area sample sizes are too small for reliable design-based direct estimation, model‐based methods borrow information across areas and incorporate auxiliary information
to improve the precision of estimates \citep{rao:molina:15}. 
In most low- and middle-income countries (LMICs), census or administrative microdata are limited. Routinely collected household surveys serve as the only reliable sources for many key health and demographic indicators. Demographic and Health Surveys (DHS) \citep{dhs_program, corsi2012dhs} and the Multiple Indicator Cluster Surveys (MICS) \citep{unicef_mics, khan2019mics} are two prominent examples that have been implemented in many countries over several decades. Small area estimation methods have been developed and applied in this setting for outcomes including poverty, education, child mortality, fertility, vaccination, and HIV \citep[see, e.g., ][]{chi2022microestimates,mercer:etal:15, li:etal:19, saha2023small,dong:wakefield:21,wu2026small}. 

SAE models are generally categorized into either area-level or unit-level models. The most widely used area-level model is the Fay-Herriot model \citep{fay:herriot:79}, where the direct estimates and their associated estimates of the sampling variances are used as response data in a smoothing model. In contrast, unit-level approaches model individual survey responses, often via generalized linear mixed models. 
Empirical comparisons reveal no universally dominant approach: performance depends on many factors, such as the survey design, spatial dependence structure, and data sparsity \citep{wakefield2025twocultures, dong2026toward}.
Model evaluation and selection therefore remains a fundamental challenge in SAE. Existing approaches are either restricted to comparing models within shared likelihood families or require external benchmark data or assumptions about the true data-generating mechanism. There is a critical need for a unified, model-agnostic framework for comparing candidate estimators on a common and scientifically meaningful scale.

In this paper, we develop a design-based cross-validation framework to address this challenge of model comparison in SAE, motivated by the problem of mapping subnational prevalence of health indicators using complex household surveys. The framework evaluates small area estimators against held-out, design-based direct estimates constructed from sample splits that respect the survey design, enabling model-agnostic comparison across area-level and unit-level approaches. We derive a decomposition of conditional mean squared errors for small area estimators that reveals a consistent cross-validation score. We show that while finite-sample score differences contain an unidentifiable bias component, this component can be rigorously bounded, yielding a principled threshold for ranking competing models. We also show that popular alternatives, such as leave-one-area-out cross-validation, produce systematic bias in model selection. 

The proposed framework is evaluated on the task of estimating the proportion of women aged 15--49 who are literate in Zambia using the 2024 Zambia Demographic and Health Survey. Female literacy is a key indicator of human capital and gender equity, and a direct target under Sustainable Development Goal 4.6, which calls for universal literacy by 2030 \citep{sdgsWeb}. Despite national-level progress, subnational disparities in female literacy remain pronounced across sub-Saharan Africa \citep{baten2021educational}. Mapping such inequalities at fine geographic resolution is of direct policy relevance, since educational programming and resource allocation are typically planned at the province or district level. Prior work on subnational literacy estimation has explored both high-resolution spatial geostatistical methods \citep{bosco2017} and area-level SAE models \citep{wu2026small}, yet a principled basis for selecting among competing modeling approaches has been lacking.

The rest of the paper is organized as follows. 
Section \ref{sec:data} introduces the 2024 Zambia DHS, the target estimand, and the candidate SAE models under comparison.
Section \ref{sec:background} reviews current approaches to evaluating small area estimators in the literature and discusses their limitations.
Section \ref{sec:method} develops the proposed scoring and cross-validation framework. 
Section \ref{sec:sim} reports simulation results and Section \ref{sec:real-data} applies the methodology to compare estimators of female literacy rate in Zambia at both the province and district level. Section \ref{sec:discuss} concludes with a discussion of future directions. All proofs are deferred to the Supplementary Materials, which also include code and data to replicate our findings.

\section{Motivating Example: the 2024 Zambia DHS}\label{sec:data}
\subsection{Survey Design}\label{sec:survey-design}
Similar to most DHS surveys, the 2024 Zambia DHS adopts the stratified multi-stage cluster design, where the primary sampling units (PSUs) are enumeration areas (EAs), or clusters, and within each PSU, the secondary sampling units (SSUs) are households. The survey is based on a sampling frame from the 2022 Census of Population and Housing of the Republic of Zambia. The population comprises $36,770$ clusters. Zambia has $10$ provinces, which are further divided into $115$ districts.
In the first stage, each of the provinces was stratified by urban and rural residence, yielding $20$ strata in total. Within each stratum, clusters were selected using probability proportional to size (PPS) sampling, resulting in $545$ clusters in total. In the second stage, within each selected cluster, households were selected through equal-probability systematic sampling, with a fixed take of $25$ households per cluster, yielding a target sample of $13,625$ households. 
For eligible women aged 15--49 in the sampled households, the DHS collects information on a broad range of demographic, reproductive, maternal, and child health indicators, along with related socioeconomic and behavioral characteristics. In this study, we consider estimating the proportion of women who are literate, where literacy is defined as women who attended schooling higher than the secondary level or who can read a whole sentence or part of a sentence.

\subsection{Direct Estimation}\label{sec:direct}
For a finite population with $N$ units and $M$ non-overlapping domains, let $U_i \subset \{1, 2, \ldots, N\}$, $i = 1, \ldots, M$, denote the index set of units in the $i$-th domain. In this paper, we consider domains defined by administrative areas. 
Let $y_j \in \{0, 1\}$ denote the response value from the $j$-th unit and $s[j] \in \{1, ..., M\}$ denote which subpopulation the $j$-th unit belongs to. We focus on binary outcomes here, though the method extends to other variable types. The primary quantity of interest is the area-level prevalence $\btheta = (\theta_1, ..., \theta_M)$.

From the realized survey sample $S \subset U_1 \cup \cdots \cup U_M$, with their design weights $w_j = 1 / \pi_j$, where $\pi_j$ is the probability that the $j$-th unit is sampled, the H{\'a}jek estimator \citep{hajek:71} for subpopulation prevalence is given by
\begin{equation}
\hat \theta_i^{w} = \frac{\sum_{j \in S\cap U_i}w_j y_j}{\sum_{j \in S\cap U_i}w_j}, \;\;\;\; i = 1, ..., M.
\end{equation} 

\subsection{Area- and Unit-Level SAE Models}\label{sec:candidate-models-intro}
We consider two widely adopted model classes for estimating subnational prevalence. At the area level, we consider Fay-Herriot models with
\begin{eqnarray}
\hat \phi_i^{w} \mid \theta_i \sim N(\logit(\theta_i), \hat V_i), \;\;\;\;
\logit(\theta_i) = \alpha + \bx_i^T\bbeta + u_i, \;\;\;\;i = 1, ..., M, \label{eq:FH}
\end{eqnarray}
where $\hat \phi_i^{w} = \logit(\hat \theta_i^{w})$ is the transformed direct estimate, and $\hat V_i$ is the associated asymptotic variance estimate, $\bbeta$ are fixed effects for covariates $\bx_i$, and $u_i$ are random effects.

At the unit level, a Bernoulli likelihood for the binary survey responses is equivalent to a cluster-level binomial model. For a survey consisting of $C$ clusters, we let $c[j]$ be the cluster index of the $j$-th individual and $i[c]$ be the area index of the $c$-th cluster. Let $Y_{c}$ and $n_c$ denote the number of positive responses and the number of sampled units from the $c$-th cluster, respectively, i.e., $Y_{c} = \sum_{j: c[j] = c} y_{j}$ and $n_{c} = \sum_{j: c[j] = c} 1$. We consider a cluster-level model with a beta-binomial likelihood \citep{dong:wakefield:21, dong2026toward,wakefield2025twocultures},
\begin{align}
Y_c \mid p_c &\sim \mbox{BetaBinomial}(n_c, p_{c}, d) .
\end{align}
In this formulation, $p_c$ is the latent prevalence for the $c$-th cluster, and $d$ is the overdispersion parameter that captures the additional within-cluster variation. 
The cluster-level prevalence is then linked to area-level prevalence via
$
\logit(p_{c}) = \alpha + \bx_i^T\bbeta +u_{i[c]},
$
with random effects $u_i$. For both model classes, we take the point estimate $\hat{\theta}_i$ to be the posterior mean of $\theta_i$. 

Figure~\ref{fig:literacymap} shows a comparison of different estimators at both province and district levels, with detailed specification described in Section \ref{sec:sim-adm1}. While both SAE models target the same quantity of interest $\btheta$, the modeling assumptions and likelihood structures are completely different, making likelihood-based model comparison infeasible. 

\begin{figure}[!htb]
    \centering
    \includegraphics[width=1 \textwidth]{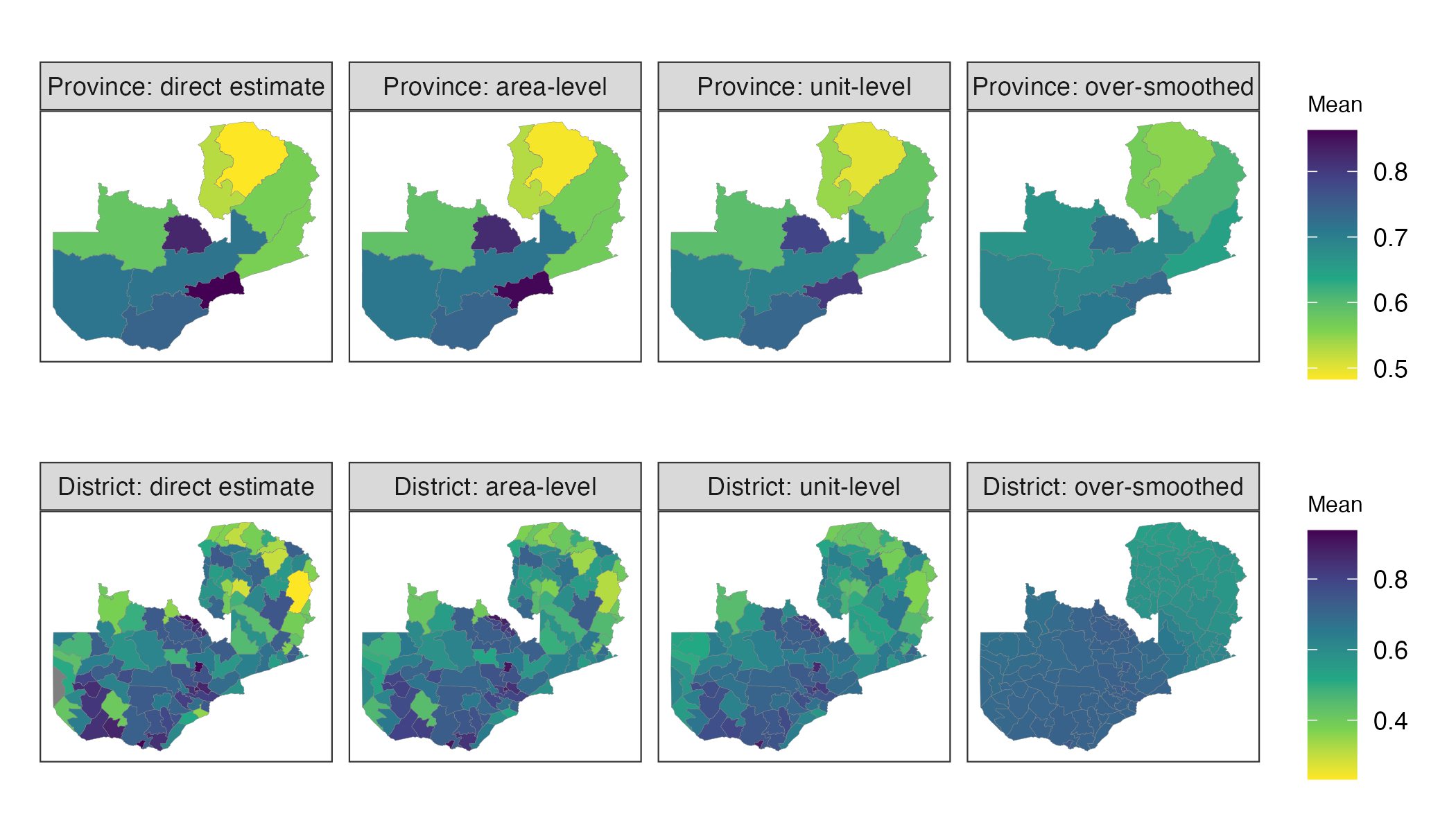}
    \caption{Province- and district-level estimates of the female literacy rate
    under direct estimation and the three candidate SAE models considered in Section \ref{sec:real-data}. 
    }  
    \label{fig:literacymap}
\end{figure}

\section{Current Approaches to Evaluating SAE Models} \label{sec:background}
Existing model evaluation approaches fall broadly into three categories: benchmark-based, likelihood-based, and empirical simulation. None of the approaches provides a generally applicable, model-agnostic framework for comparing estimators. The first class of methods involves evaluating candidate estimators in terms of their agreement with some benchmark values. External validation against census data is widely regarded as the gold standard for model assessment \citep{merfeld2024small, merfeld2022combining}, and when a census is not available, validation against estimates from independent surveys has been proposed \citep{brown2001evaluation, dorfman2018towards}. However, external census data or surveys are rarely available at the required spatial resolution. 
Alternatively, one may first aggregate SAE estimates to coarser levels, e.g., a single national estimate, and then compare them with external benchmark values. 
While useful for falsifying estimators, aggregated evaluation changes the target estimand completely. Estimators that perform well at the aggregated level can still be significantly biased at the finer levels, especially with over-smoothed models. 
When no external data are available, direct estimates from the same sample are also sometimes used to assess small area estimators \citep{lahiri2019evaluation}. However, such evaluation suffers from double-dipping and systematically favors models that track the noise in direct estimates. 

 



Traditional model assessment and validation tools, including goodness-of-fit tests \citep{rao1981analysis, rao1984chi, thomas1987small, lumley2014tests, lumley2017fitting}, information criteria \citep{fabrizi2013design, lumley2015aic}, and generalization error measures \citep{holbrook2020estimating}, have been developed for regression models that incorporate survey weights, with adjustments to account for complex sampling designs. Similar approaches for comparing Bayesian SAE models using measures such as DIC and WAIC have also been adopted in the literature \citep[see, e.g., ][]{trevisani2017comparison,gomez2008comparison,franco2023combining}.
These methods are generally limited to comparing models that share a comparable likelihood formulation. Consequently, they cannot be used to compare models with fundamentally different constructions, e.g., between area-level and unit-level models.

Finally, empirical simulation approaches are also commonly used in the SAE literature \citep{bradley2015multivariate, janicki2022bayesian}. These approaches perturb observed direct estimates by adding noise, and then evaluate models fitted to the synthetic direct estimates against the original estimates. When some individual-level microdata from the population exist, design-based simulation may also be carried out by generating repeated surveys based on a synthetic population dataset \citep{alfons2010simulation, wieczorek2013small}. Both approaches evaluate model performance for a hypothetical population rather than the specific observed dataset at hand, thus implicitly changing the target of inference.

The cross-validation framework developed in Section~\ref{sec:method} retains the model-agnostic advantage of benchmark-based approaches while resolving the lack of benchmark datasets. It is closely related to the concurrent work of \citet{kawano2026dt}, in which model evaluation approaches are developed by splitting direct estimates into independent copies for model fitting and evaluation, using data thinning \citep{neufeld2024data, dharamshi2025generalized}. \citet{kawano2026dt} is restricted to comparing area-level models with a Gaussian likelihood that enables data thinning. By sample splitting at the level of individual survey responses rather than area-level summaries, the 
framework in this paper enables comparison across models with  
different likelihood structures and extends to any estimators, including those constructed with machine learning or other black-box algorithms.




\section{Cross-Validating Small Area Estimators}\label{sec:method}
\subsection{Setup and Notation}
Building on the notation introduced in Section \ref{sec:direct}, we denote a generic partition of a sample $S$ into a training set $A$ and a validation set $B = S\setminus A$, and write $\hat\btheta_A$ and $\hat\btheta_B$ for the estimators based on the two subsamples. Specifically for $K$-fold cross-validation, we use $\hat\theta^{(k)}$ and $\hat\theta^{(-k)}$ to denote the estimators based on the $k$-th held-out dataset and the remaining samples, respectively. Lastly, $\mathrm{E}[\cdot]$ denotes the expectation over the survey sampling  distribution of $S$, and $\mathrm{E}[\cdot \mid S]$ denotes the expectation over the resampling distribution of the partitions given the realized sample $S$.

\subsection{Scoring Small Area Estimators} \label{sec:score}
For most SAE applications, a natural metric to score the estimator is the weighted squared error loss,
$
L(\hat \btheta, \btheta)= \sum_{i=1}^M q_i (\hat \theta_i - \theta_i)^2,
$
for some weights $q_i$. The choice of the aggregation weights is contextual. They may represent population proportions by area, or simply be set to $1/M$ to weight all areas equally. Evaluating $\hat\btheta$ based on the loss function instead of model-induced likelihood enables model comparison across different model classes. Similar principles have been adopted by \citet{kuh2024using} in the context of multilevel regression and poststratification.

This squared error loss is closely related to the mean squared error (MSE), a topic extensively studied in the SAE literature \cite[see, e.g., ][]{Prasad:1990, lahiri1995robust,rivest2000conditional,butar2003measures, rao:molina:15}. Specifically, our ultimate target of inference for model $\MM$ is the weighted average of area-specific conditional MSEs,
\begin{equation} \label{eq:mse}
    \mbox{MSE}^{\MM} =  \sum_{i=1}^M q_i \EE[(\hat \theta_{i}^{\MM} - \theta_i)^2],
\end{equation}
where $\btheta$ is considered fixed and unknown, and the expectation is taken over the distribution of sampling units. The term $\EE[(\hat \theta_{i}^{\MM} - \theta_i)^2]$ is usually referred to as the conditional MSE in the literature \citep{rivest2000conditional,lohr2009jackknife}. The majority of the theoretical results on estimating MSE, however, focus on the unconditional case where the expectation is further taken over the distribution of $\btheta$, assuming area-level Gaussian random effects \citep{Prasad:1990, datta2000unified}. To the best of our knowledge, the score proposed in this paper is the first cross-validation-based estimator of the conditional MSE without imposing distributional assumptions on the random effects.



\subsection{Stratified $K$-fold Cross-Validation} \label{sec:scv}
Following the multi-stage design, 
we randomly assign SSUs, i.e., households, within each cluster to one of the $K$ folds. This creates $K$ folds of data, each containing all clusters in the survey but only $1/K$ of the households within each cluster. The survey design is preserved within each fold, and the selection probability only needs to be adjusted by $1/K$, and survey weights remain proportionally unchanged. 
Given a partition, the naive estimator of the squared error loss is
\begin{equation} \label{eq:mse-naive}
\widehat{\mathrm{score}}_{\mathrm{naive}}^{\MM} = \sum_{i=1}^M  \frac{q_i}{K} \sum_{k=1}^{K}
\left(\hat{\theta}_i^{\MM, (-k)}-\hat{\theta}_i^{w,(k)}\right)^2 ,
\end{equation}
where \(\hat{\theta}_i^{\MM, (-k)}\) denotes the model-based estimate of $\theta_i$ obtained using data excluding the \(k\)-th fold, and \(\hat{\theta}_i^{w,(k)}\) denotes the direct estimate for the \(i\)-th area computed from the \(k\)-th held-out fold. 
The partitioning step can also be repeated and averaged to stabilize the estimated score.

Stratified cross-validation under complex survey designs has been investigated in the survey literature, usually for the purpose of variable selection in regression models \citep{iparragirre2023variable, wieczorek2022k}. Our goal is different here. Instead of evaluating individual-level prediction, we compare model-based estimators of population prevalence with the direct estimates from the held-out samples. The design-aware sample splitting is key to approximating the target estimand, which is detailed next.

To understand the properties of CV scores, it is important to first note that MSE estimators based on sample splitting target a different estimand at a reduced sample size \citep{bates2024cross}. Specifically, an MSE estimator based on a subset of samples $A \subset S$ targets $\text{MSE}_{A}^{\mathcal{M}} = \sum_{i=1}^M q_i \mathbb{E}[(\hat\theta_{A,i}^{\mathcal{M}} - \theta_i)^2]$, rather than $\text{MSE}^{\mathcal{M}}$ in Equation (\ref{eq:mse}). 
In $K$-fold cross-validation, set $A$ has approximately $(K-1)/K$ of the size of $S$.
To differentiate the two targets, we refer to $\text{MSE}^{\mathcal{M}}$ as the oracle full-sample MSE and $\text{MSE}_{A}^{\mathcal{M}}$ as the oracle training MSE. A similar gap in the target estimand was also observed in \citet{kawano2026dt}. The relationship between the two estimands, however, is not a simple function of sample size difference, due to both the complex survey design and smoothing models. For the purpose of model comparison, our goal is to estimate model ranking under the original design. It is therefore desirable to keep the training sets close to the full-sample distribution, in order to reduce this inestimable gap. For multi-stage designs, one may choose to stratify at different stages to form the partitions. In the Supplementary Materials, we include one alternative splitting strategy where clusters are assigned to the $K$ folds instead of households , and show that such an approach performs similarly to cross-validation of SSUs when sample size is large, but is more biased when sample size is small, due to a larger gap between the training oracle and full-sample oracle estimand.  


\subsection{An Adjusted CV Score}
We now turn to estimators of $\text{MSE}_{A}$ using sample splitting. The naive CV score in Equation (\ref{eq:mse-naive}) is a biased estimator for $\text{MSE}_{A}$, because the held-out direct estimate carries sampling variation. Theorem \ref{thm1} identifies the bias structure for sample splitting estimators in general.

\begin{theorem}[Estimation bias due to sample splitting]
\label{thm1}
For the $i$-th area, consider a partition of $S$ into two sets $A$ and $B$ and model $\MM$, and assume that the direct estimates $\mathbb{E}[\hat{\theta}_{B,i}^{w}] = \theta_i$. The naive MSE estimator, $\EE[(\hat\theta_{A,i}^{\MM} - \hat\theta^{w}_{B, i})^2 \mid S]$, is biased for the oracle training MSE, with 
\begin{equation}
\EE[\EE[(\hat\theta_{A,i}^{\MM} - \hat\theta^{w}_{B, i})^2 \mid S]] - \EE[(\hat\theta_{A,i}^{\MM} - \theta_i)^2] = 
\EE(v_i + b_i^2 - 2c_i^{\MM} - 2b_id_i^{\MM})
\end{equation}
where the outer expectation is taken over the distribution of $S$, and the bias components are
\begin{align}
    v_i &= \var(\hat\theta_{B,i}^w \mid S),\\
    b_i &= \EE(\hat\theta_{B,i}^w - \theta_i \mid S),\\
    d_i^{\MM} &= \EE(\hat\theta_{A,i}^{\MM}- \theta_i \mid S),\\
    c_i^{\MM} &= \cov(\hat\theta_{A, i}^{\MM}, \hat\theta_{B,i}^w \mid S).
\end{align}
\end{theorem}

\begin{remark}
\label{remark1}
When the index sets $A$ and $B$ are from two independent surveys, instead of a random partition of the same survey, the bias reduces to $\EE(-v_i)$ and an unbiased estimator exists in the form of
$\EE[(\hat\theta_{A,i}^{\MM} - \hat\theta^{w}_{B, i})^2] - v_i$.
\end{remark}

The conditional biases in Theorem \ref{thm1}, $b_i$ and $d_i^\MM$, are not identifiable without knowing the true $\btheta$. It is important to note that while the direct estimator is design unbiased, splitting a survey into two sets results in the realized error being carried over to all resulting subsamples. Thus $b_i$ is nonzero given a finite sample $S$. This leads to an unidentifiable bias when evaluating the difference of MSEs across models. Theorem \ref{thm2} provides an adjusted score and characterizes the remainder bias term for cross-validation estimators. 

\begin{theorem}[Adjusted CV score]
\label{thm2}
For the $i$-th area, define the fold-averaged means
$\bar{\theta}^w_i = K^{-1}\sum_{k = 1}^K \hat\theta_i^{w,(k)}$
and
$\bar{\theta}^\MM_i = K^{-1}\sum_{k=1}^K \hat\theta_i^{ \MM,(-k)}$,
the empirical variance of the held-out direct estimates
\begin{equation}
    \hat v_i = \frac{1}{K}\sum_{k = 1}^K(\hat\theta^{w,(k)}_i - \bar{\theta}^w_i)^2,
\end{equation}
and the empirical covariance between the held-out direct estimates and the training-set model estimates
\begin{equation}
    \hat c_i^\MM = \frac{1}{K}\sum_{k=1}^K (\hat\theta^{w,(k)}_i - \bar{\theta}^w_i)(\hat\theta^{ \MM,(-k)}_i  - \bar{\theta}^\MM_i).
\end{equation}
The adjusted CV score for model $\MM$ is
\begin{equation}\label{eq:score}
     \widehat{\mathrm{score}}_i^{\MM} =
     \frac{1}{K}\sum_{k=1}^{K} \left(\hat{\theta}_i^{\MM, (-k)}-\hat{\theta}_i^{w,(k)}\right)^2
     -\hat v_i + 2\hat c_i^{\MM}.
\end{equation}
For two models $\MM_1$ and $\MM_2$, the expected score difference satisfies
\begin{equation}\label{eq:errordiff}
\EE(\widehat{\mathrm{score}}_i^{\MM_1} - \widehat{\mathrm{score}}_i^{\MM_2}) =
\EE\big[(\hat\theta_{A,i}^{\MM_1} - \theta_{i})^2 -
(\hat\theta_{A,i}^{\MM_2} - \theta_{ i})^2 \big] + e_i,
\end{equation}
where the expectation is over the sampling distribution of $S$. The remainder term is
\begin{equation}\label{eq:error}
e_i = -2\,\cov(b_i,\; d_i^{\MM_1} - d_i^{\MM_2}).
\end{equation}
where $b_i = \EE(\hat\theta_{i}^{w,(k)} - \theta_i \mid S)$ and $d_i^{\MM} = \EE(\hat\theta_{i}^{ \MM,(-k)}- \theta_i \mid S)$, and the covariance is over the distribution of samples $S$.
\end{theorem}

Of the two correction terms, $v_i$ is a constant across models and $c_i^{\MM}$ is analogous to the covariance penalty in \citet{efron2004estimation}, reflecting the shared survey structures between the training and validation samples. Theorem \ref{thm2} also shows that the residual bias in the score difference is driven by the covariance between the conditional bias of the direct estimate in the validation set, $b_i$, and the model difference in the training set, $d_i^{\MM_1} - d_i^{\MM_2}$. The bias is large when, across hypothetical replicate surveys, the sampling error in the direct estimator for area $i$ is systematically associated with the difference in model biases. This occurs, for instance, when one model tracks the direct estimates more closely than the other.
When there is only a single sample $S$, the bias $e_i$ cannot be identified. Therefore, when two models have similar conditional MSE,
the oracle MSE difference in Equation \eqref{eq:errordiff} is small
and the remainder $e_i$ may dominate the score difference,
making the comparison  uninformative. To account for this, we establish a finite-sample bound for $|e_i|$ in Theorem \ref{thm3}, and show that $e_i \to 0$ as data accumulate in Theorem \ref{thm:consistency}.

\begin{theorem}[Approximate error bound for adjusted CV score difference]
\label{thm3}
For each area $i$, let $n_i = |S \cap U_i|$ denote the within-area
sample size. Under regularity conditions in the Supplementary Materials, $e_i$ can be bounded by 
\begin{equation}
      |e_i|  \leq 2\sqrt{2 \var(\hat\theta_{S, i}^w)\Big(\var(\hat\theta_{S, i}^{\MM_1}) +\var(\hat\theta_{S, i}^{\MM_2})\Big)}+ o(n_i^{-1})
\end{equation}
where the subscript $S$ denotes models fitted on the full sample $S$, and $o(n_i^{-1})$ collects linearization remainders.
Therefore, a plug-in estimator of the bound is
\begin{equation}\label{eq:bound}
  t_i=2\sqrt{2\widehat{\var}(\hat\theta_{S, i}^w)\Big(\widehat{\var}(\hat\theta_{S, i}^{\MM_1}) +\widehat{\var}(\hat\theta_{S, i}^{\MM_2})\Big)},
\end{equation}
where $\widehat{\var}(\hat\theta_{S, i}^w)$ is the standard design-based variance estimator for the direct estimate based on the full sample, e.g., by linearization of the H\'{a}jek estimator. For Bayesian models, the posterior variance 
$\mathrm{var}_{\mathrm{post}}(\theta^{\mathcal{M}}_{S,i})$ provides a 
plug-in estimator for $\widehat{\mathrm{var}}(\hat{\theta}^{\mathcal{M}}_{S,i})$.
\end{theorem}

For the estimation of the aggregated squared error, the aggregated adjusted CV score is $\widehat{\mbox{score}}_q^{\MM}=\sum_{i=1}^M q_i \widehat{\mbox{score}}_i^{\MM}$, with the aggregated bias
\(e_q = \sum_{i=1}^M q_i e_i\),
satisfying
\begin{equation}\label{eq:agg-bd}
|e_q|
\le
\sum_{i=1}^M q_i |e_i|
\le
\sum_{i=1}^M q_i t_i
\end{equation}
Thus $t_q = \sum_{i=1}^M q_i t_i$ serves as a threshold for deciding whether an observed aggregated score difference is large enough to support a ranking between $\MM_1$ and $\MM_2$. The full procedure for $K$-fold cross-validation of small area estimators is summarized in Algorithm \ref{alg:cv}. 

\begin{algorithm}[htb]
\caption{$K$-fold Stratified cross-validation for comparing small area estimators}\label{alg:cv}
\begin{algorithmic} 
\Require Multi-stage stratified survey sample $S$, candidate models $\MM_1$ and $\MM_2$, aggregation weights $q_1, ..., q_M$.
\begin{enumerate}
    \item Randomly partition SSUs within each cluster into $K$ folds and adjust for survey weights accordingly.
    \item \textbf{for} $k = 1, \ldots, K$ \textbf{do}
     \item \hspace*{1em}  Compute the $k$-th held-out direct estimates $\hat\btheta^{w,(k)}$.
    \item \hspace*{1em}  Fit $\MM_1$ and $\MM_2$ on the remaining samples to obtain $\hat\btheta^{ \MM_1,(-k)}$ and $\hat\btheta^{ \MM_2,(-k)}$. 
 \item Compute adjusted CV scores for $\MM_1$ and $\MM_2$ using Equation (\ref{eq:score}).
 \item Compute the error bounds $t_i$ using Equation (\ref{eq:bound}) and the aggregated bound $t_q = \sum_{i=1}^M q_i t_i$.
 \item  \textbf{If} $|\mbox{score}^{\MM_1} - \mbox{score}^{\MM_2}| > t_q$  \textbf{then} return the model with the smaller score.
 \item \textbf{Else} comparison is inconclusive.
\end{enumerate}
\end{algorithmic}
\end{algorithm}

\begin{remark}
    We focus on $K$-fold cross-validation here, but the MSE decomposition in Theorem \ref{thm1} and the error bound in Theorem \ref{thm3} apply to a broad range of sample-splitting schemes, as long as the two sets $A$ and $B$ are both representative of the full sample $S$. The cross-validation procedure in Algorithm \ref{alg:cv} can also be directly extended to repeated sample splitting with averaged estimators. A special case of such an averaged estimator with $K = 2$ is provided in the Supplementary Materials.
\end{remark}

Finally, the preceding results provide finite-sample tools for model comparison. A natural question is whether the remainder $e_i$ vanishes as data accumulate. Theorem \ref{thm:consistency} shows that the adjusted score is  asymptotically consistent for the conditional MSE difference under mild conditions.

\begin{theorem}[Consistency of the adjusted score]\label{thm:consistency}
With $K$ and $M$ fixed, assume $\var(\hat\theta_{S,i}^w) \to 0$ as $n \to \infty$ and 
$\var(\hat\theta_{S,i}^{\MM})$ is bounded for each model $\MM$, 
then $e_i \to 0$  for each area $i = 1, \ldots, M$.
Consequently,
$
\EE\big(\widehat{\mathrm{score}}_q^{\MM_1} 
      - \widehat{\mathrm{score}}_q^{\MM_2}\big)
- \big(\mathrm{MSE}_A^{\MM_1} - \mathrm{MSE}_A^{\MM_2}\big)
\;\to\; 0.
$
\end{theorem}

\subsection{The Failure of Leave-One-Area-Out (LOAO) Cross-Validation}\label{sec:loao}
To illustrate the pitfalls of sample splitting strategies without maintaining the original survey design, we examine a popular validation scheme in the literature using leave-one-area-out (LOAO) cross-validation. Under LOAO, the model is fitted on data from \(M-1\) areas and then used to predict the prevalence in the held-out area. Let \(\hat{\theta}_i^{(-i)}\) denote the predicted prevalence for the $i$-th area using the rest of the data. An analogous score evaluating the squared error loss is 
\[
\widehat{\mathrm{score}}_{\mathrm{LOAO}}=\sum_{i=1}^M q_i\left(\hat{\theta}_i^{(-i)}-\hat{\theta}_i^{w}\right)^2.
\]
LOAO cross-validation has been used heuristically in the SAE literature \citep{steorts2020smoothing, wakefield2025twocultures}. For unit-level models, where individual observations are nested within areas, LOAO is a form of block cross-validation. For area-level models, LOAO corresponds directly to the standard leave-one-out cross-validation (LOO-CV) and can be evaluated using Pareto smoothed importance sampling without refitting the model $M$ times \citep{vehtari2016bayesian}. The computational efficiency makes it a widely used diagnostic tool to assess Bayesian area-level models. 

Although conceptually simple and straightforward to implement, LOAO is an extrapolation estimator targeting a different estimand, $\EE[(\hat\theta_i^{(-i)} - \theta_i)^2]$. Proposition \ref{prop:loao-decomp} shows that it can lead to reversed rankings when comparing smoothing models in typical SAE settings.
\begin{proposition}\label{prop:loao-decomp}
For the $i$-th area, the difference between the extrapolation and smoothing MSE is
\[
\EE[(\hat\theta_i^{(-i)} - \theta_i)^2] - \EE[(\hat\theta_i - \theta_i)^2] = \Big(\var(\hat\theta_i^{(-i)}) - \var(\hat\theta_i)\Big) + \Big((\Delta_i^{(-i)})^2 - (\Delta_i)^2\Big),
\]
where $\Delta_i = \EE(\hat\theta_i- \theta_i)$ and $\Delta_i^{(-i)} = \EE(\hat\theta_i^{(-i)} - \theta_i)$ are the conditional biases of the smoothing and extrapolation estimators, and the expectations are taken over the sampling distribution conditional on $\btheta$.
\end{proposition}

The proof of Proposition \ref{prop:loao-decomp} follows from standard MSE decomposition. The gap between the two MSE estimands manifests differently across models. It is easy to see that for a strongly over-smoothed model, the gap is close to $0$ as the model estimates are not sensitive to dropping data from a single area. However, for a well-calibrated model, both the estimation variance and bias of the estimator for the $i$-th area are expected to increase after removing all data from the $i$-th area, thus making the extrapolation MSE larger than the smoothing MSE. Therefore, LOAO systematically penalizes flexible models more than over-smoothed models. We numerically illustrate this bias in Section \ref{sec:sim}.

\section{Simulation Analysis} \label{sec:sim}



\subsection{Design-Based Synthetic Survey Simulation} \label{sec:sim-design}
We first consider a design-based simulation study that mimics the 2024 Zambia DHS. Since the 2024 Zambia DHS was based on the sampling frame from the 2022 census, we constructed a 300m $\times$ 300m population grid using the 2022 population surface from WorldPop \citep{worldpop2018}. To construct the sampling frame, we ranked pixels by population within each province and retained the most populated pixels to
match the provincial EA counts from the 2024 DHS sampling frame, yielding $36{,}770$ gridded pixels. 
We treat these pixels as the master sampling frame of our study. For simplicity, we ignore urban/rural status of the pixels in the simulation and consider synthetic surveys  stratified by province. 

We construct the synthetic population in this simulation study using smoothed district-level estimates of women's literacy rate, obtained by fitting a Fay-Herriot model on the 2024 Zambia DHS. This introduces realistic within-area heterogeneity in our simulated data.
For a given cluster $c$ with population $N_c$ in a district with prevalence $p$, we generate event count $Y_c$
from a Beta-Binomial distribution with mean $N_c p$ and variance $N_c p(1-p)(1 + (N_c - 1)d)$, where $d = 0.00001$ introduces additional overdispersion within the district. In this simulated population, at the district level, $N_c$ varies between 12 and 462 with the average around 65, resulting in excess binomial variance between $0.011\%$ and $0.461\%$. The simulated finite population prevalence in the $i$-th area is then $\theta_{i} = \sum_{c \in i} Y_c / \sum_{c \in i} N_c$. 

\subsection{Province-Level Modeling}\label{sec:sim-adm1}
We start with modeling prevalence at the $10$ province level. Given the synthetic population, we draw 50 replicate surveys using a two-stage stratified design, where $50$ clusters are sampled in each stratum and $30$ households are sampled in each cluster. We consider the following three candidate models that differ in their model class and prior specification. 

\begin{itemize} 
\item $\MM_1$: Unit-level Beta-Binomial model with independent random effects $u_i \stackrel{iid}{\sim} N(0, \sigma_u^2)$, and a $PC (1, 0.01)$ hyperprior on $\sigma_u$.
\item $\MM_2$: Area-level Fay-Herriot model with the same priors for random effects as in $\mathcal{M}_1$.
\item $\MM_3$: Same model as in $\MM_1$, but with a highly informative $PC (0.01, 0.01)$ on $\sigma_u$.
\end{itemize} 
For all three models, we adopt the penalised complexity (PC) prior \citep{simpson:etal:17}, where the $PC(U, \alpha)$ prior for the standard deviation $\sigma$ specifies that $Pr(\sigma >U) = \alpha$. For both $\MM_1$ and $\MM_2$, the prior correspond to default priors commonly used in practice \citep{dong2026toward}. 
Figure~\ref{fig:sim_int} displays, for a single replicate, point estimates and 90\% uncertainty intervals produced by the candidate models. Direct estimation, $\MM_1$, and $\MM_2$ all produce estimates close to the population truth across all provinces. By contrast, $\MM_3$ shows excessive shrinkage, though with overlapping interval estimates relative to $\MM_1$ and $\MM_2$.

\begin{figure}[htb]
    \centering
    \includegraphics[width=1\textwidth]{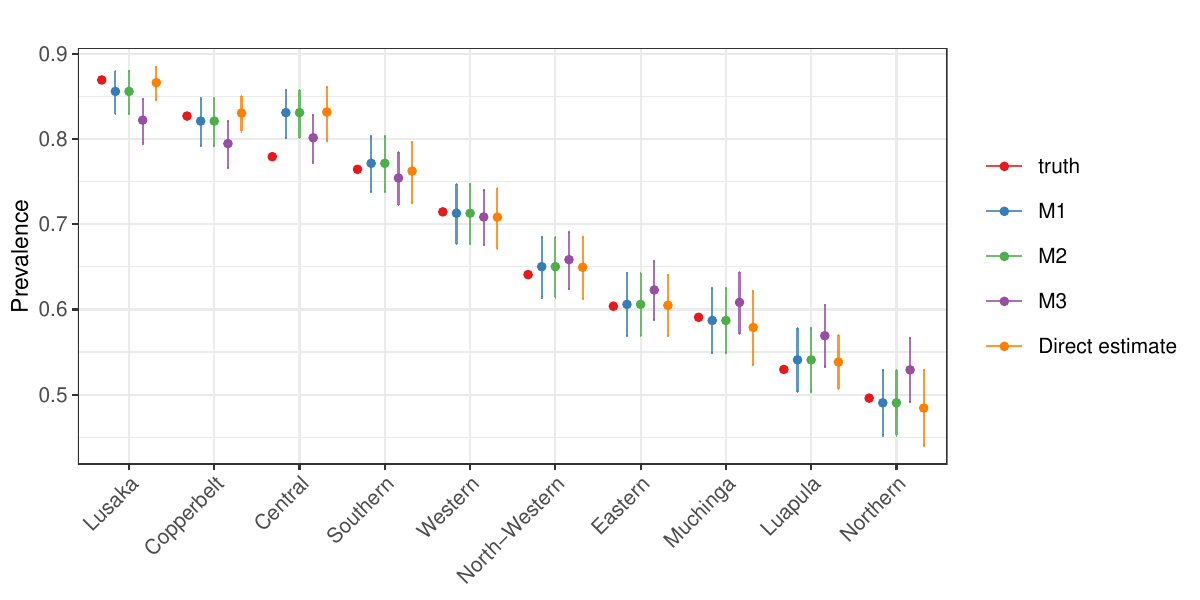}
\caption{Province-level prevalence estimates and 90\% credible intervals with the synthetic population truth for a single simulated survey. 
}  \label{fig:sim_int}
\end{figure}



 We computed the adjusted cross-validation scores of all models using $5$-fold stratified cross-validation. Figure~\ref{fig:bd_diff_nntl} summarizes the adjusted CV score comparisons of the three models across the $50$ synthetic surveys. When comparing $\MM_3$ with $\MM_1$, the adjusted CV score consistently identifies that $\MM_3$ has larger MSE and the score differences all exceed the error bound, $t_q$, providing decisive evidence in favor of $\MM_1$. The comparison between $\MM_2$ and $\MM_1$ is more difficult as the oracle differences are much smaller in scale. The adjusted CV score differences mostly have the same sign as the oracle difference, but in all replicates, the estimated score differences were between $\pm t_q$, thus suggesting the comparison between $\MM_2$ and $\MM_1$ is inconclusive, as the available information in the data does not overcome the unidentifiable bias components of the conditional MSE.

\begin{figure}[t]
    \centering
    \includegraphics[width=.8\textwidth]{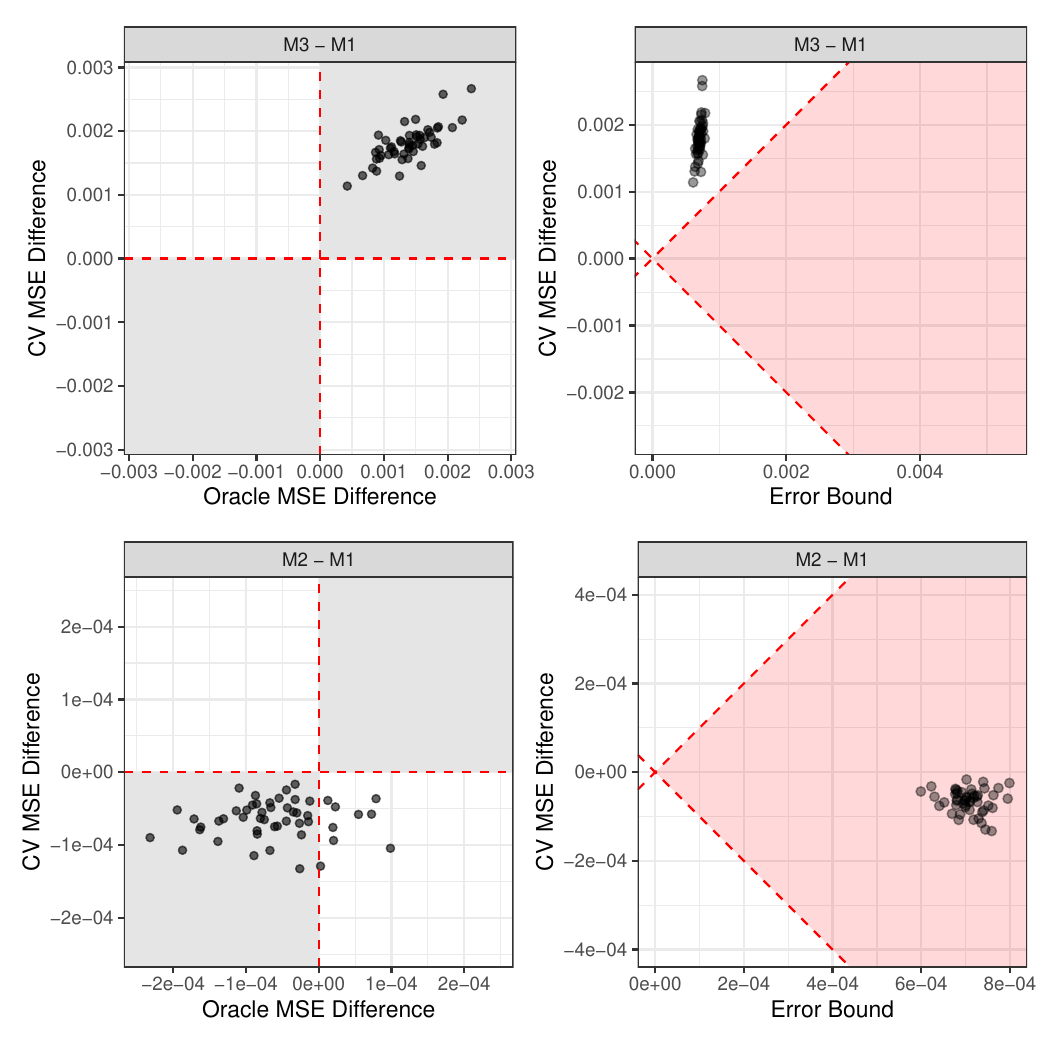}

  \caption{Province-level model comparison.
\textit{Left:} Adjusted CV score differences versus oracle full-sample MSE differences for $\mathcal{M}_3 - \mathcal{M}_1$ (top) and $\mathcal{M}_2 - \mathcal{M}_1$ (bottom). Each point represents one synthetic survey replicate. Points in the first and third quadrants indicate correct model ranking. \textit{Right:} Adjusted CV score differences versus the error bound. The red shaded region is the inconclusive region where  score differences are too small to support ranking of the two models.}  \label{fig:bd_diff_nntl}
\end{figure}

To gain additional insight on the model performance, we also examine the area-specific scores in  
Figure~\ref{fig:bd_diff_subntl}. For the comparison between $\MM_3$ and $\MM_1$, the over-smoothed $\MM_3$ departs noticeably from the other models in Copperbelt, Luapula, Lusaka, and Northern provinces, producing absolute score differences that are well above the error bound.
For comparing $\MM_2$ with $\MM_1$, area-specific score differences still remain tightly clustered near zero, with all points falling within the inconclusive region, except for Lusaka on a few occasions. 

\begin{figure}[htb]
    \centering
    \includegraphics[width=.9\textwidth]{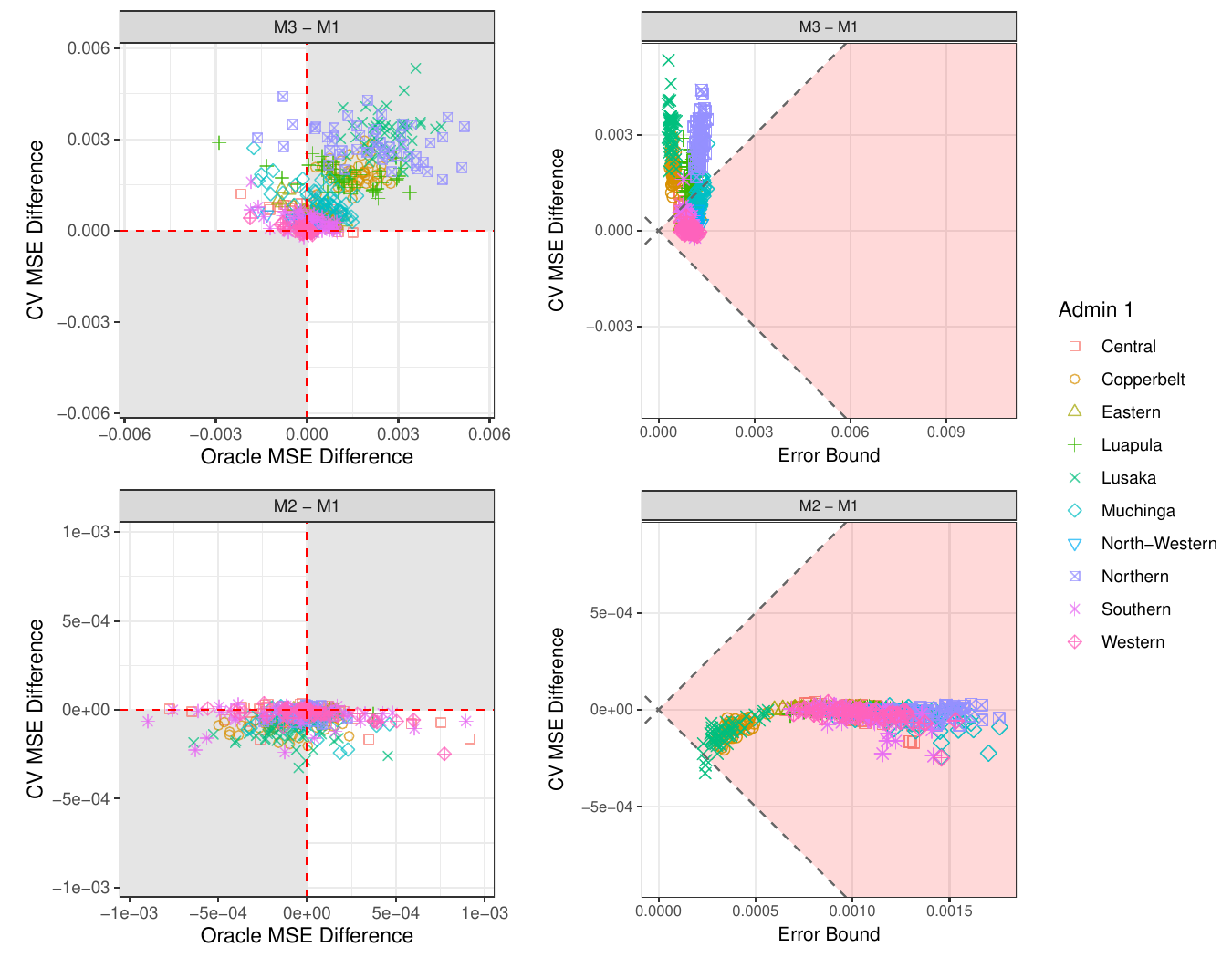}
 \caption{\textit{Left:} Area-level adjusted CV score differences versus oracle error differences for $\mathcal{M}_3 - \mathcal{M}_1$ (top) and $\mathcal{M}_2 - \mathcal{M}_1$ (bottom). \textit{Right:}  Adjusted CV score differences versus the error bound $t_i$ by  provinces.  Each point represents one simulated dataset, colored by province. }
    \label{fig:bd_diff_subntl}
\end{figure}

In summary, these results demonstrate that the CV score reliably ranks the models for the performance of subnational prevalence estimation, while the error bound $t_q$ provides the principled threshold for separating genuine differences from noise-dominated ones. Additional diagnostics and visualizations are provided in Supplementary Materials.


We also compare the LOAO strategy in Figure~\ref{fig:loao} for $\MM_1$ and $\MM_3$ under LOAO validation. As discussed in Section \ref{sec:loao}, LOAO penalizes the over-smoothed model less. Indeed, in this comparison, the procedure reached the opposite conclusion and ranked $\MM_3$ as the preferred model in every replicate, confirming the bias described in Section \ref{sec:loao}. 


\begin{figure}[htb]
    \centering
    \includegraphics[width=\textwidth]{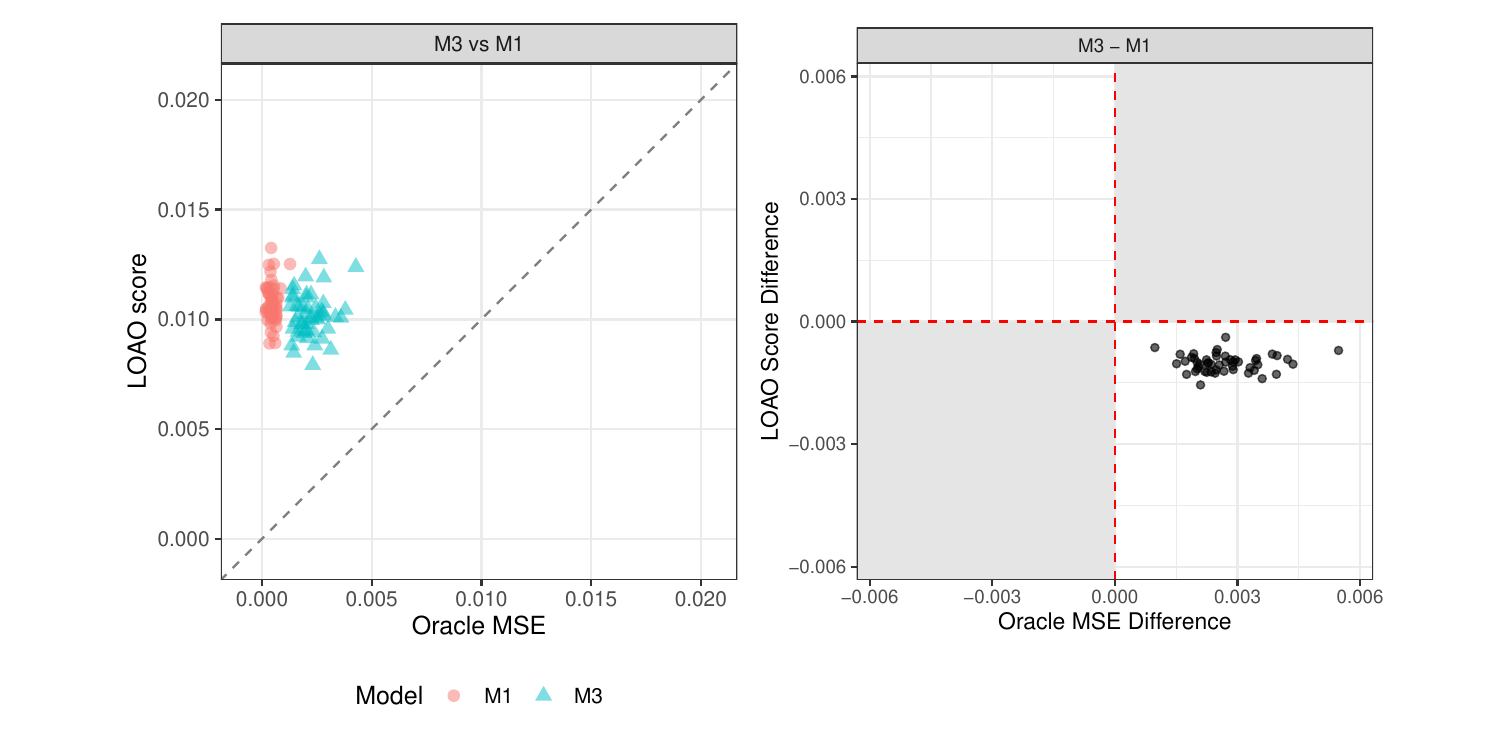}
\caption{Comparison of $\MM_1$ and $\MM_3$ under LOAO validation across 50 simulation replicates. \textit{Left}: LOAO scores versus oracle full-sample MSE for the two models. \textit{Right}: LOAO score difference, $\widehat{\mathrm{score}}_{\mathrm{LOAO}}^{\MM_3}-\widehat{\mathrm{score}}_{\mathrm{LOAO}}^{\MM_1}$, versus oracle full-sample MSE difference, $\mathrm{MSE}_{\mathrm{oracle}}^{\MM_3}-\mathrm{MSE}_{\mathrm{oracle}}^{\MM_1}$.}
    \label{fig:loao}
\end{figure}

In this simulation setting, the relative performance of the candidate models remains unchanged between the full-sample and oracle training MSE. When sample sizes are smaller, the two estimands may lead to different model rankings, which we investigate in the Supplementary Materials. We also compare different sample-splitting strategies and the choice of $K$ in the Supplementary Materials.

\subsection{District-Level modeling}\label{sec:sim-adm2}
We now turn to the finer geographical resolution with $115$ districts. 
We generate synthetic surveys using the same population described in Section \ref{sec:sim-adm1} and sample $60$ clusters per stratum and $30$ households per cluster. When considering subnational areas below the stratification level, direct estimates or the associated variance estimator can be unavailable due to data sparsity. Therefore, when evaluating direct estimates on the validation set, we consider only districts with at least one sampled cluster. Across the $50$ replicates, the number of districts with microdata for direct estimates ranges from $89$ to $99$, with a median of $94$. To account for the stronger spatial dependence at the district level, we replace the iid prior for $u_i$ in all three models in Section \ref{sec:sim-adm1} with a spatially structured BYM2 prior \citep{besag:york:mollie:91, riebler:etal:16}. A $PC(1, 0.01)$ prior is used for the marginal standard deviation parameter of the random effects in $\MM_1$ and $\MM_2$, and a $PC(0.01, 0.05)$ prior is used for $\MM_3$ to induce stronger spatial smoothing. In all models, we put a $PC(0.5, 2/3)$ prior on the proportion of variance that is spatial. Figure~\ref{fig:LIT-3060-ad2ad2_bd_diff_nntl_w} summarizes the comparison results across all $50$ replicates, and the conclusions are similar to that in Section \ref{sec:sim-adm1}, where $\MM_1$ is correctly favored over $\MM_3$ while $\MM_1$ and $\MM_2$ comparison is inconclusive for all replicate surveys.


\begin{figure}[htb]
    \centering
    \includegraphics[width=.6\textwidth]{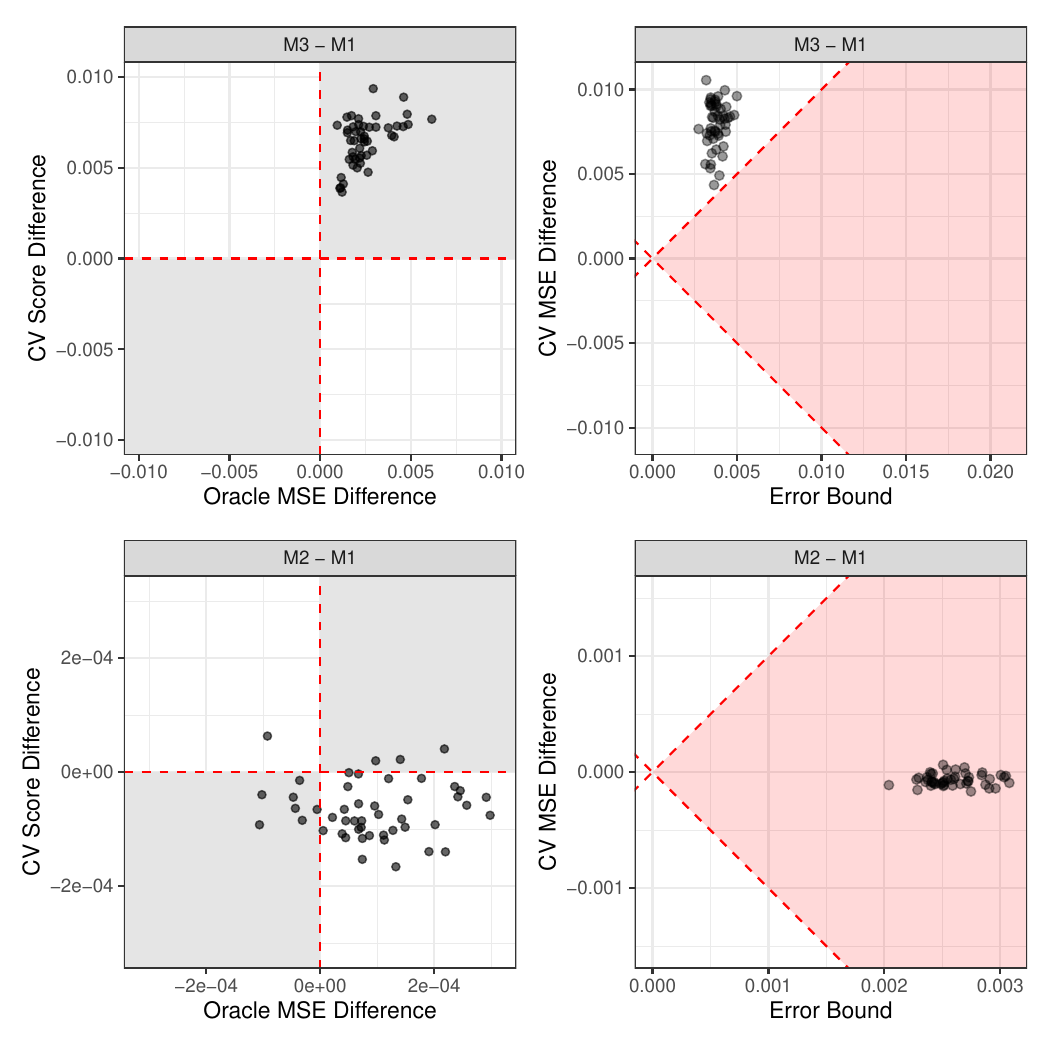}
     \caption{District-level model comparison.
\textit{Left:} Adjusted CV score differences versus oracle full-sample MSE differences for $\mathcal{M}_3 - \mathcal{M}_1$ (top) and $\mathcal{M}_2 - \mathcal{M}_1$ (bottom). Each point represents one synthetic survey replicate. Points in the first and third quadrants indicate correct model ranking. \textit{Right:} Adjusted CV score differences versus the error bound. The red shaded region is the inconclusive region where  score differences are too small to support ranking of the two models.}  
    \label{fig:LIT-3060-ad2ad2_bd_diff_nntl_w}
\end{figure}



\section{Analysis of the 2024 Zambia DHS}\label{sec:real-data}

We now return to the 2024 Zambia DHS survey and compare models estimating the percentage of women who are literate. The stratified cross-validation was conducted with the $20$ strata defined by province crossed with urban/rural residence. We consider the same candidate models in Section \ref{sec:sim-adm1} at the province level and those in Section \ref{sec:sim-adm2} at the district level. In both cases,  $\MM_1$ and $\MM_2$ and their corresponding prior choices have been routinely used as default models for analyzing DHS data \citep[see, e.g.,][]{dong2026toward}, and $\MM_3$ leads to over-smoothed estimators due to the more concentrated prior choice.

Province-level results from 5-fold cross-validation are reported in Figure~\ref{fig:real_data_eg}. Estimates from $\MM_3$ are heavily shrunk to the overall mean. Between $\MM_1$ and $\MM_2$, we observe slightly more shrinkage in $\MM_1$, and the shrinkage is more obvious in Lusaka and Copperbelt. The cross-validation procedure again identifies that $\MM_3$ has larger aggregated score that exceeds the estimated threshold, compared with $\MM_1$. The difference between $\MM_1$ and $\MM_2$, on the other hand, is substantially smaller and within the inconclusive region. 

\begin{figure}[!htb]
    \centering
        \includegraphics[width=.48\textwidth]{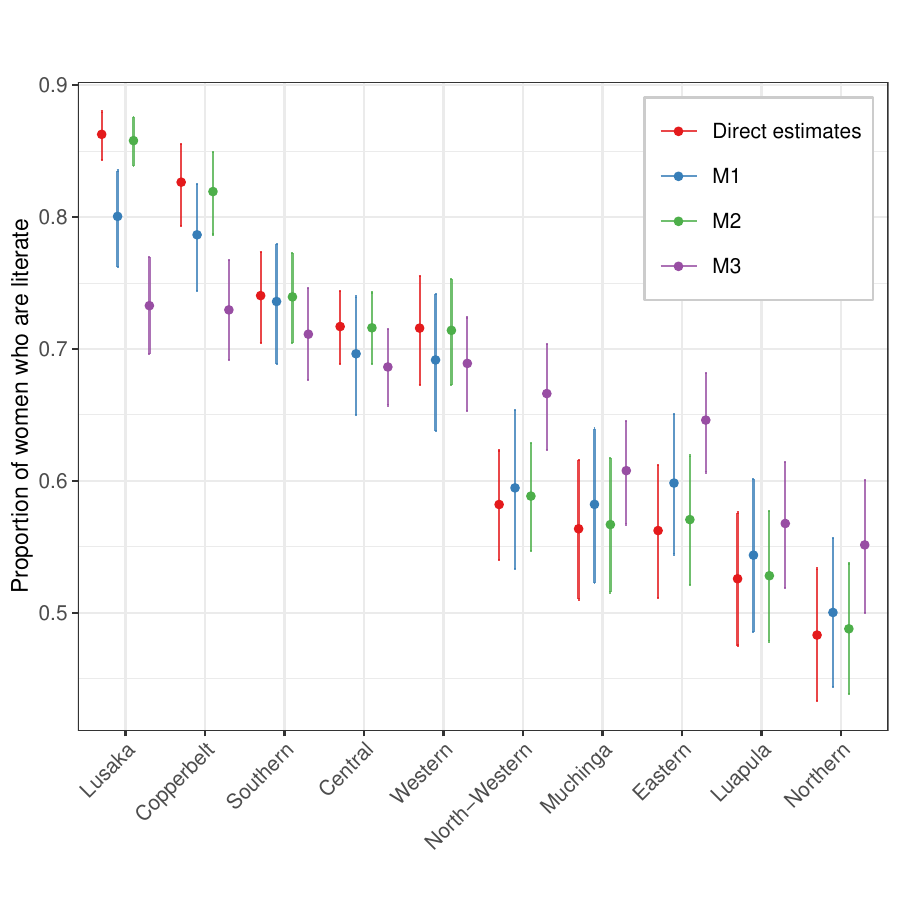}
        \includegraphics[width=.48\textwidth]{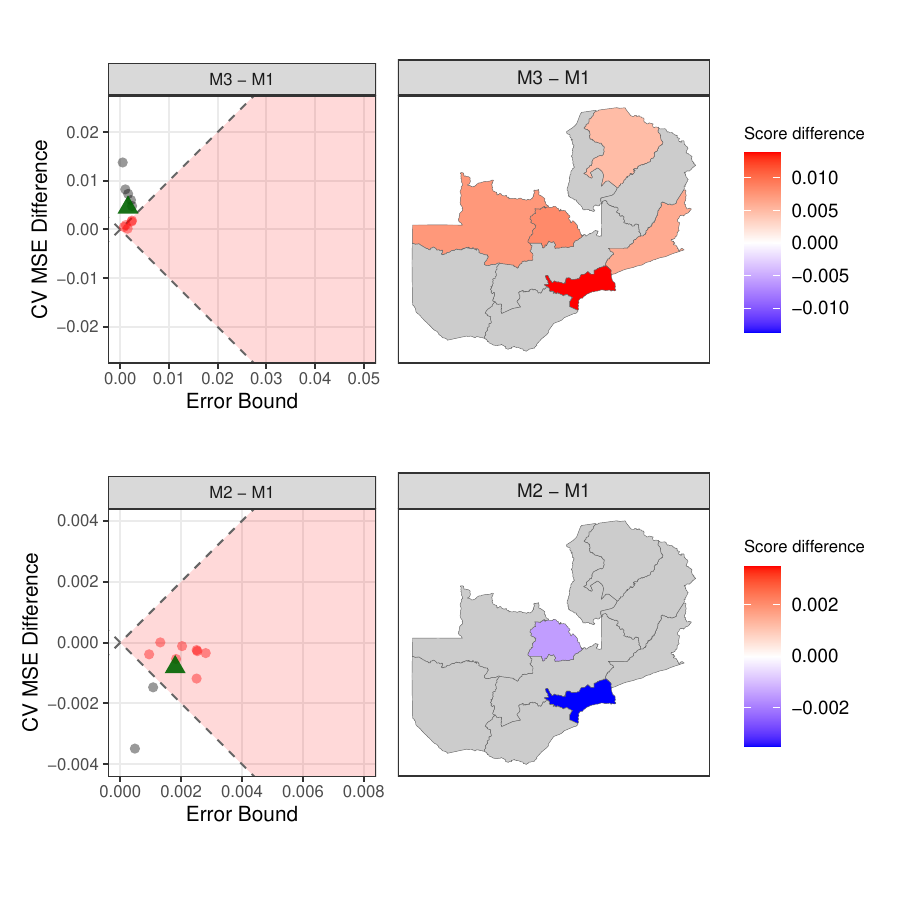}
    \caption{Comparison of province-level female literacy rate estimates . \textit{Left}: Point estimates and 90\% uncertainty intervals. \textit{Right}: $5$-fold adjusted CV score comparison by province. The top row compares $\MM_3$ with $\MM_1$, and the bottom row compares $\MM_2$ with $\MM_1$. In each scatter plot, points denote the 10 provinces and the triangle denotes the aggregated result. The maps show the associated area-level score differences, with grey area denoting inconclusive ranking. 
    }
    \label{fig:real_data_eg}
\end{figure}
A close examination of the area-specific scores in Figure \ref{fig:real_data_eg} provides additional insight into the behavior of the models. When comparing $\MM_1$ and $\MM_3$, the score differences are negative and exceed the area-specific threshold in magnitude across all ten provinces, yielding a uniform preference for $\MM_1$ over $\MM_3$. 
When comparing $\MM_1$ and $\MM_2$, while the overall score differences are smaller than the potential bias in magnitude, there are two provinces, Lusaka and Copperbelt, with score differences exceeding their area-specific thresholds. Lusaka and Copperbelt are the two provinces with the highest literacy rate and estimators from $\MM_1$ exhibit more shrinkage than that from $\MM_2$, as can be observed in Figure \ref{fig:real_data_eg}. The cross-validation procedure identifies such over-shrinkage and suggests $\MM_2$ is preferred in these two provinces, although the reduction in MSE from these two areas is not sufficient to overcome the unidentifiable bias in other provinces when considering a single summary score. For practitioners, we note that while this analysis suggests an inconclusive comparison between $\MM_1$ and $\MM_2$, a different set of aggregation weights may lead to $\MM_2$ being preferred over $\MM_1$. Thus, it is important to choose the appropriate aggregation weights based on the specific task at hand when a single decision is needed.


At the district level, one district has no sampled clusters and thus is omitted from the analysis. Similar to the province-level models, $\MM_1$ and $\MM_2$ are close to each other, while $\MM_3$ is over-smoothed. Figure~\ref{fig:ad2_real_diff} shows the district-level comparison, with similar model comparison results as the province level analysis. 
Compared to province-level models, district-level estimation uses less data and therefore has larger uncertainty. This leads to larger and more conservative error bounds and thus fewer areas with conclusive comparisons. In both analyses, the cross-validation procedure was able to identify the over-smoothed model $\MM_3$ as being inferior to the more reasonable options, while also showing that the Fay-Herriot model  $\MM_2$ is preferred over unit-level model $\MM_1$ in a subset of areas, though such differences are not large enough for conclusive ranking when considering all areas together.  Nevertheless, it is worth noting that this finding challenges the conventional view that Fay-Herriot models necessarily break down under severe data sparsity, leaving unit-level or geostatistical models as the only viable options. A key distinction between the two models here is that $\MM_2$ incorporates the survey design through the weighted direct estimates, whereas the unit-level model $\MM_1$ does not account for the stratification by urbanicity. This suggests that accounting for the design can matter more than model flexibility, and that the advantage of unit-level models at fine spatial resolutions may only materialize when those models properly account for the survey design. 


\begin{figure}
    \centering
    \includegraphics[width=0.7\linewidth]{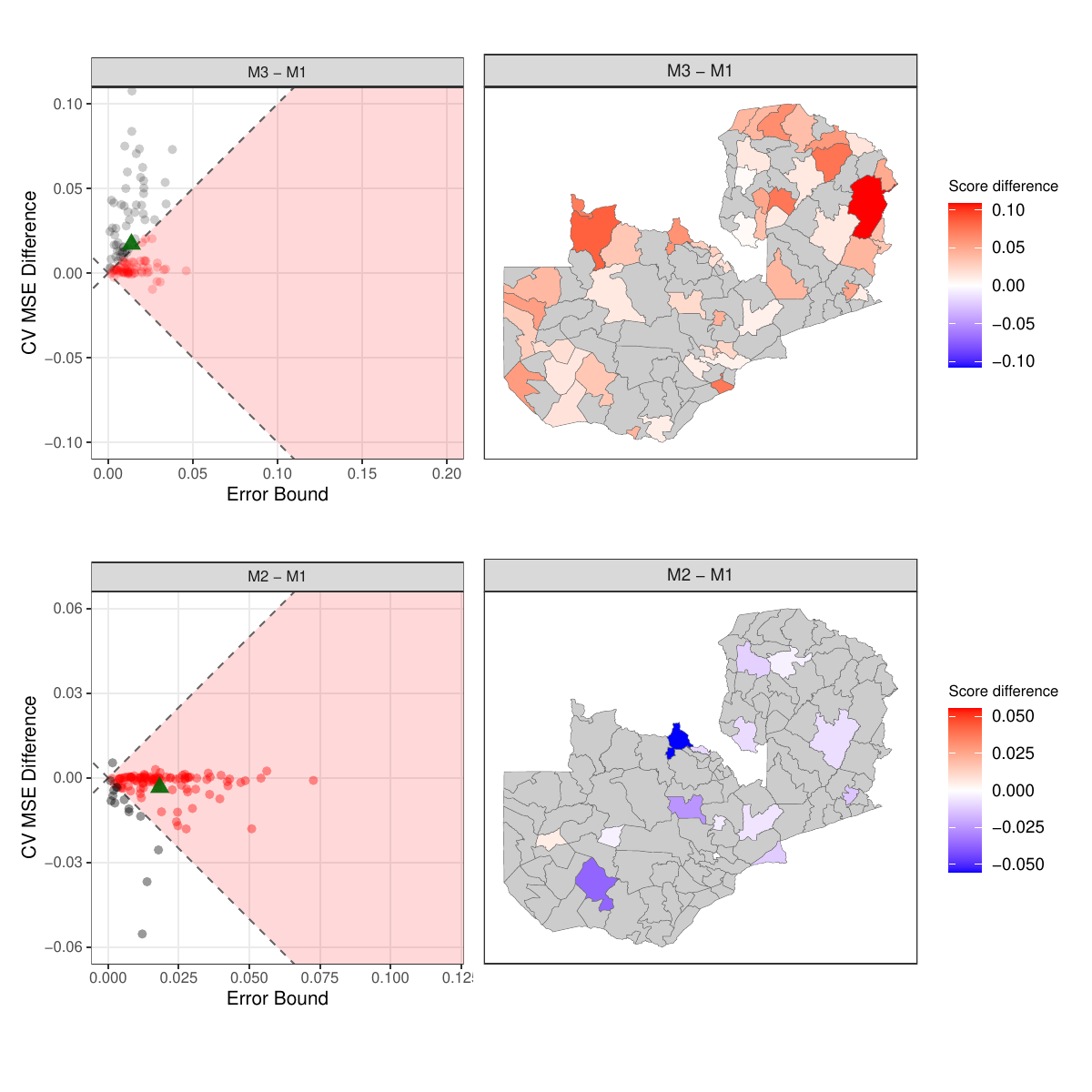}
    \caption{Comparison of district-level female literacy rate estimates using the 2024 Zambia DHS. The top row compares $\MM_3$ with $\MM_1$, and the bottom row compares $\MM_2$ with $\MM_1$. In each scatter plot, points denote the 115 districts and the triangle denotes the aggregated result. The maps show the associated area-level score differences, with grey area denoting inconclusive ranking.} 
    \label{fig:ad2_real_diff}
\end{figure}

\section{Discussion}\label{sec:discuss}
In this article, we proposed a design-based $K$-fold cross-validation framework for comparing small area estimators under complex survey designs, with application to subnational estimation of female literacy rates using the 2024 Zambia DHS. The framework splits household-level survey data within clusters to construct held-out design-based direct estimates, evaluates model-based estimates against these held-out estimates, and corrects for sampling variation through an adjusted CV score. We established a decomposition of the cross-validated squared error identifying an unidentifiable bias component, and derived a finite-sample bound on this remainder that provides a principled threshold for ranking competing models. We also demonstrated that the
widely used leave-one-area-out cross-validation targets out-of-sample extrapolation rather than in-sample smoothing accuracy and can reverse model rankings relative to the true performance ordering. Applied to the Zambia DHS at both the province and district level, the proposed framework reliably distinguished the over-smoothed model from the well-calibrated alternatives, while providing useful insights into the comparison between two sets of similar estimates from the Fay-Herriot and unit-level models.

The proposed cross-validation framework leaves several methodological questions open. First, although the aggregate score is defined for general area weights, the choice of the weighting scheme itself deserves further study. In some applications, one may prefer to use variance to aggregate the scores, so that discrepancies in areas with different uncertainty receive different emphasis. Such choices may be attractive when model comparison is intended to reflect not only average squared error, but also performance in areas where the estimator is most informative. At the same time, the theoretical consequences of variance-based weighting are not yet clear, especially when the weights depend on the fitted model through posterior uncertainty estimates.

The squared error loss considered in this paper evaluates only point estimates.  When full predictive distributions are available, one may instead use the continuous ranked probability score (CRPS), where  the full posterior predictive CDF is compared against held-out direct estimates. Establishing a similar decomposition and finite-sample error bound for CRPS can be a promising future research direction.

As with all sample-splitting methods, the success of the proposed procedure depends on the availability of sufficient data. When the target of inference is at the stratification level, e.g., provinces of Zambia, the impact of data sparsity is usually not severe. While we demonstrate that the procedure can also be useful at the level of $115$ districts, it remains an open question how to deal with even more extreme data sparsity, e.g., when no samples are collected in many areas. 

Finally, the optimal choice of $K$ in the $K$-fold cross-validation is another important question relevant to practitioners. As we demonstrate empirically in the Supplementary Materials, a larger $K$ reduces the gap between $\text{MSE}_A$ and the oracle full-sample MSE, at the cost of validation power due to the small sample sizes used for validation. Similar observations were made in \citet{kawano2026dt}.
We leave these topics for future research.

\bibliographystyle{apalike}
\bibliography{cv_sae}

\clearpage
\appendix

\section{Proofs}

Consider splitting a sample $S$ into two disjoint sets $A$ and $B = S\setminus A$ following one of the sample splitting schemes discussed in Section \ref{sec:method}. For a sequence of random variables $\{X_n\}$ and a sequence of constants $\{a_n\}$, $X_n = O_p(a_n)$ indicates that $X_n/a_n$ is bounded in probability and $X_n = o_p(a_n)$ indicates that $X_n/a_n$ converges in probability to zero. $O(\cdot)$ and $o(\cdot)$ denote the corresponding deterministic asymptotic bounds. 

\subsection{Proof of Theorem \ref{thm1}}
\begin{proof}
 Consider the conditional expectation of the first term:
\begin{align*}
    \EE[(\hat\theta_{A,i}^{\MM} - \hat\theta^{w}_{B, i})^2 \mid S] &= \EE[(\hat\theta_{A,i}^{\MM} - \theta_i - (\hat\theta^{w}_{B, i} - \theta_i))^2 \mid S] \\
    &= \EE[(\hat\theta_{A,i}^{\MM} - \theta_i)^2 \mid S] + \EE[(\hat\theta^{w}_{B, i} - \theta_i)^2 \mid S] - 2\EE[(\hat\theta_{A,i}^{\MM} - \theta_i)(\hat\theta^{w}_{B, i} - \theta_i) \mid S].
\end{align*}
Expanding the last two terms yields:
\begin{align*}
    \EE[(\hat\theta^{w}_{B, i} - \theta_i)^2 \mid S] &= \var(\hat\theta_{B,i}^w \mid S) + (\EE[\hat\theta^{w}_{B, i} - \theta_i \mid S])^2 = v_i + b_i^2, \\
    \EE[(\hat\theta_{A,i}^{\MM} - \theta_i)(\hat\theta^{w}_{B, i} - \theta_i) \mid S] &= \cov(\hat\theta_{A, i}^{\MM}, \hat\theta_{B,i}^w \mid S) + \EE(\hat\theta_{A,i}^{\MM} - \theta_i \mid S)\EE(\hat\theta_{B,i}^w - \theta_i \mid S) = c_i^{\MM} + d_i^{\MM}b_i.
\end{align*}
Taking the expectation over $S$ yields the result in the theorem.
\end{proof}

\subsection{Proof of Theorem \ref{thm2}}

\begin{proof}
From Theorem~\ref{thm1}, $\EE(\mbox{score}_S^{\MM}) = \EE[(\hat\theta_{A,i}^{\MM} - \theta_i)^2] + \EE(v_i + b_i^2) - 2\EE(b_id_i^{\MM})$. 
The difference is $\EE[(\hat\theta_{A,i}^{\MM_1} - \theta_i)^2] - \EE[(\hat\theta_{A,i}^{\MM_2} - \theta_i)^2] - 2\EE[b_i(d_i^{\MM_1} - d_i^{\MM_2})]$.
Since $\hat\theta^w$ is design unbiased, $\EE(b_i) = 0$. Thus, $\EE[b_i(d_i^{\MM_1} - d_i^{\MM_2})] = \cov(b_i, d_i^{\MM_1} - d_i^{\MM_2})$.
\end{proof}

%

\subsection{Proof of Theorem \ref{thm3}}
\label{s-thm3}

\begin{condition}[Linearization of the model-based estimator difference]
\label{cond:diff}
For each area $i$ and each pair of models $(\MM_1, \MM_2)$, there exist
bounded functions $\Delta\psi_{i,j}$ of $j \in S$, fixed given $S$, such that
for any subsample $A \subseteq S$ with $n_A = |A|$,
\begin{equation}
\hat\theta^{\MM_1}_{A,i} - \hat\theta^{\MM_2}_{A,i}
= \frac{1}{n_A}\sum_{j \in A}\Delta\psi_{i,j} + \Delta r_{A,i},
\qquad \Delta r_{A,i} = O_p(n_A^{-1}).
\end{equation}
\end{condition}

\begin{condition}[Linearization of the H\'ajek direct estimator]
\label{cond:hajek}
For each area $i$, there exist bounded functions $\psi^w_{i,j}$, defined on
units $j \in U_i$ and not depending on the random partition, such that for
any subset $B \subseteq S$ with $n_{B,i} = |B \cap U_i|$,
\begin{equation}
\hat\theta^w_{B,i} - \theta_i
= \frac{1}{n_{B,i}} \sum_{j \in B \cap U_i} \psi^w_{i,j} + r^w_{B,i},
\qquad
r^w_{B,i} = O_p(n_{B,i}^{-1}).
\label{eq:lin-hajek}
\end{equation}
\end{condition}

\begin{lemma}[Conditional expectation of a subsample-based difference under linearization]
\label{lem:subsamp-diff}
Let $S$ be a fixed sample of size $n = |S|$, and let $A \subset S$ be a
random subsample under a mechanism that assigns each unit $j \in S$ to $A$
with equal inclusion probability $\pi_A \in (0,1)$, with $n_A = \pi_A n$
under balanced partitioning. For area $i$ and two estimators
$\hat\theta^{\MM_1}_{A,i}$ and $\hat\theta^{\MM_2}_{A,i}$ of $\theta_i$,
suppose their difference admits the linearization
\begin{equation}
\hat\theta^{\MM_1}_{A,i} - \hat\theta^{\MM_2}_{A,i}
= \frac{1}{n_A}\sum_{j \in A}\Delta\psi_{i,j} + \Delta r_{A,i},
\qquad \Delta r_{A,i} = O_p(n_A^{-1}),
\label{eq:lem-diff-lin}
\end{equation}
where $\Delta\psi_{i,j}$ are bounded functions of unit $j$, fixed given $S$,
and the same expansion holds at $A = S$ with remainder
$\Delta r_{S,i} = O_p(n^{-1})$. Then
\begin{equation}
E\!\left(\hat\theta^{\MM_1}_{A,i} - \hat\theta^{\MM_2}_{A,i} \;\big|\; S\right)
= \hat\theta^{\MM_1}_{S,i} - \hat\theta^{\MM_2}_{S,i} + O_p(n^{-1}).
\label{eq:lem-diff-conclusion}
\end{equation}
\end{lemma}

\begin{proof}
Taking conditional expectation of~\eqref{eq:lem-diff-lin} given $S$ and
writing the sum using indicators,
\begin{equation}
E\!\left(\hat\theta^{\MM_1}_{A,i} - \hat\theta^{\MM_2}_{A,i} \mid S\right)
= \frac{1}{n_A}\sum_{j \in S} P(j \in A \mid S)\,\Delta\psi_{i,j}
  + E[\Delta r_{A,i} \mid S],
\end{equation}
where the $\Delta\psi_{i,j}$ come out as constants since they are fixed
given $S$. Under equal inclusion probability $\pi_A$ and $n_A = \pi_A n$,
\begin{equation}
\frac{1}{n_A}\sum_{j \in S} P(j \in A \mid S)\,\Delta\psi_{i,j}
= \frac{\pi_A}{\pi_A n}\sum_{j \in S}\Delta\psi_{i,j}
= \frac{1}{n}\sum_{j \in S}\Delta\psi_{i,j}.
\label{eq:lem-diff-cancel}
\end{equation}
Applying~\eqref{eq:lem-diff-lin} at $A = S$ gives
$n^{-1}\sum_{j \in S}\Delta\psi_{i,j} = \hat\theta^{\MM_1}_{S,i} -
\hat\theta^{\MM_2}_{S,i} - \Delta r_{S,i}$. Substituting,
\begin{equation}
E\!\left(\hat\theta^{\MM_1}_{A,i} - \hat\theta^{\MM_2}_{A,i} \mid S\right)
= \hat\theta^{\MM_1}_{S,i} - \hat\theta^{\MM_2}_{S,i}
  - \Delta r_{S,i} + E[\Delta r_{A,i} \mid S].
\end{equation}

Since $n_A = \pi_A n$ with $\pi_A$ bounded away from zero,
$O_p(n_A^{-1}) = O_p(n^{-1})$. The bounded influence function ensures that
conditional expectation preserves the rate, so $E[\Delta r_{A,i} \mid S] = O_p(n^{-1})$. And $\Delta r_{S,i} = O_p(n^{-1})$, the remainder is $O_p(n^{-1})$,
yielding~\eqref{eq:lem-diff-conclusion}.
\end{proof}

\begin{remark}
The same argument applied to a single estimator $\hat\theta^w_{B,i}$ with
influence function $\psi^w_{i,j}$ supported on $U_i$ (rather than a
difference with influence function supported on all of $S$) yields
$E(\hat\theta^w_{B,i} - \theta_i \mid S)
= \hat\theta^w_{S,i} - \theta_i + O_p(n_i^{-1})$,
where $n_i = |S \cap U_i|$. We invoke this special case for the H\'ajek
direct estimator on $B$ in the proof of Theorem~\ref{thm3} below.
\end{remark}

\begin{proof} [proof of Theorem~\ref{thm3}]
Applying the Cauchy--Schwarz inequality to $e_i$ in Theorem~\ref{thm2},
\begin{equation}\label{eq:cs_exact}
|e_i| \le 2\sqrt{\var[E(\hat\theta_{B,i}^w - \theta_i \mid S)] \;
\var[E(\hat\theta_{A,i}^{\MM_1} - \hat\theta_{A,i}^{\MM_2} \mid S)]}.
\end{equation}
Applying Lemma~\ref{lem:subsamp-diff} in its single-estimator form to the
H\'ajek direct estimator on $B$ under Condition~\ref{cond:hajek},
\begin{equation}\label{eq:hajek-app}
E(\hat\theta_{B,i}^w - \theta_i \mid S)
= \hat\theta_{S,i}^w - \theta_i + O_p(n_i^{-1}).
\end{equation}
Applying Lemma~\ref{lem:subsamp-diff} directly to the model difference on
$A$ under Condition~\ref{cond:diff},
\begin{equation}\label{eq:model-app}
E\!\left(\hat\theta_{A,i}^{\MM_1} - \hat\theta_{A,i}^{\MM_2} \mid S\right)
= \hat\theta_{S,i}^{\MM_1} - \hat\theta_{S,i}^{\MM_2} + O_p(n^{-1}).
\end{equation}
Since $n^{-1} \le n_i^{-1}$, both remainders are $O_p(n_i^{-1})$. Using that
$\theta_i$ is non-random and each remainder has variance of smaller order
than the leading variance term,
\begin{align}
\var[E(\hat\theta_{B,i}^w - \theta_i \mid S)]
&= \var(\hat\theta_{S,i}^w) + o(n_i^{-1}),\\
\var[E(\hat\theta_{A,i}^{\MM_1} - \hat\theta_{A,i}^{\MM_2} \mid S)]
&= \var(\hat\theta_{S,i}^{\MM_1} - \hat\theta_{S,i}^{\MM_2}) + o(n_i^{-1}).
\end{align}
Substituting into~\eqref{eq:cs_exact} yields the leading-order bound
\begin{equation}\label{eq:approx_step}
|e_i|
\le 2\sqrt{\var(\hat\theta_{S,i}^w)\,
          \var(\hat\theta_{S,i}^{\MM_1} - \hat\theta_{S,i}^{\MM_2})}
  + o(n_i^{-1}).
\end{equation}
The variance of the difference
$\var(\hat\theta_{S,i}^{\MM_1} - \hat\theta_{S,i}^{\MM_2})$ is not directly
available from independent Bayesian fits of $\MM_1$ and $\MM_2$. To obtain
an implementable plug-in, we apply the inequality
$\var(X - Y) \le 2\var(X) + 2\var(Y)$, yielding
\begin{equation}\label{eq:bound_freq}
|e_i|
\le 2\sqrt{2\,\var(\hat\theta_{S,i}^w)\,
           [\var(\hat\theta_{S,i}^{\MM_1})
          + \var(\hat\theta_{S,i}^{\MM_2})]}
  + o(n_i^{-1}).
\end{equation}
This bound is conservative whenever the cross-model covariance
$\cov(\hat\theta_{S,i}^{\MM_1}, \hat\theta_{S,i}^{\MM_2})$ is non-negative,
which is expected when $\MM_1$ and $\MM_2$ are smoothing estimators fitted
to the same sample and targeting the same estimand.

The sampling variances on the right-hand side of~\eqref{eq:bound_freq} are
not available from a single sample. For the direct estimator, we use the
standard design-based variance estimator $\widehat{\var}(\hat\theta_{S,i}^w)$.
For each model-based estimator, we use the posterior variance
$\var_{\mathrm{post}}(\theta_i^{\MM})$ as a plug-in approximation. The
resulting plug-in threshold is
\begin{equation}\label{eq:bound_plugin}
t_i = 2\sqrt{2\,\widehat{\var}(\hat\theta_{S,i}^w)\,
            [\var_{\mathrm{post}}(\theta_i^{\MM_1})
           + \var_{\mathrm{post}}(\theta_i^{\MM_2})]}.
\end{equation}
When considering Bayesian models, the posterior variance, $\var_{\text{post}}(\theta_{S,i}^{\MM})$ provides a plug-in estimator for $\widehat{\var}(\hat\theta_{S, i}^{\MM})$, with $M$ fixed and $n_i \to \infty$, the posterior of $\theta_i$ concentrates at rate $n_i^{-1}$ and the posterior variance is asymptotically equivalent to the sampling variance of the posterior mean by Bernstein–von Mises theorem. In finite samples, empirical validation that the resulting $t$ is an upper bound for remainder $|e|$ is provided in Figure~\ref{fig:S_theo3} in the Supplementary Materials.
\end{proof}

\subsection{Proof of Theorem \ref{thm:consistency}}

\begin{proof} 
The remainder term is $e_i = -2\cov(b_i, d_i^{\MM_1} - d_i^{\MM_2})$. By Cauchy–Schwarz inequality, $|e_i| \le 2\sqrt{\var(b_i)\var(d_i^{\MM_1} - d_i^{\MM_2})}$. Under the design regularity assumptions and law of total variance, $\var(b_i) \le \var(\hat\theta_{B,i}^w) = O(n_i^{-1}) \to 0$.
For the second term, $\var(d_i^{\MM_1} - d_i^{\MM_2}) \le 2\var(d_i^{\MM_1}) + 2\var(d_i^{\MM_2})$, and $\var(d_i^{\MM}) \le \var(\hat\theta_{A,i}^{\MM})$ by the law of total variance. Thus given the assumption that $\var(\hat\theta_{A,i}^{\MM})$ is bounded,  $\var(d_i^{\MM_1} - d_i^{\MM_2})$ is also bounded. Hence $e_i \to 0$.

\end{proof}

\clearpage
\section{Additional simulation results}


To support the analysis presented in the main text, this section provides additional visual evaluation of model performance, error decomposition, and cross-validation reliability. Figure \ref{fig:S_map1_diff} begins this assessment by illustrating the spatial distribution of estimation bias across provinces, comparing direct estimates and three model-based approaches ($\mathcal{M}_1, \mathcal{M}_2, \mathcal{M}_3$) against the known population truth. Figure \ref{fig:S_score_scatter_national} compares the naive and adjusted CV scores against the two oracle MSE estimands. Figure \ref{fig:S_theo3} validates the approximated error bounds against the true remainder terms. Finally, Figure \ref{fig:S_score_hist_sub_h} disaggregates the CV scores and shows a distribution of the area-specific scores across the ten provinces, which clearly illustrates which provinces drive the differences in the scores.
\begin{figure}[!ht]
    \centering
    \includegraphics[width=.5\textwidth]{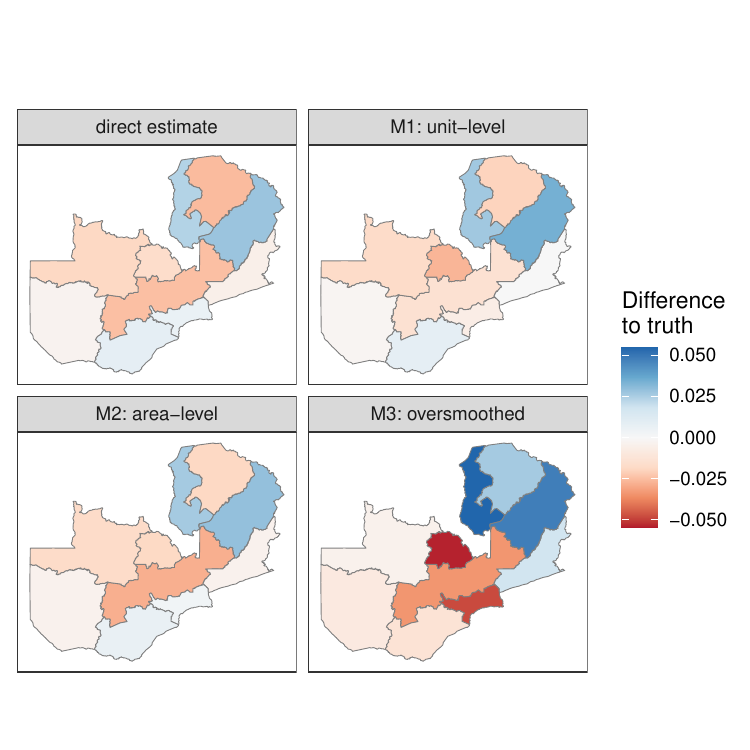}
    \caption{Province-level maps of difference of population truth with direct estimator and three model-based estimators ($\mathcal{M}_1$, $\mathcal{M}_2$, $\mathcal{M}_3$), for a single simulation replicate under the moderate sample size setting.}
    \label{fig:S_map1_diff}
\end{figure}

\begin{figure}
    \centering
    \includegraphics[width=0.6\linewidth]{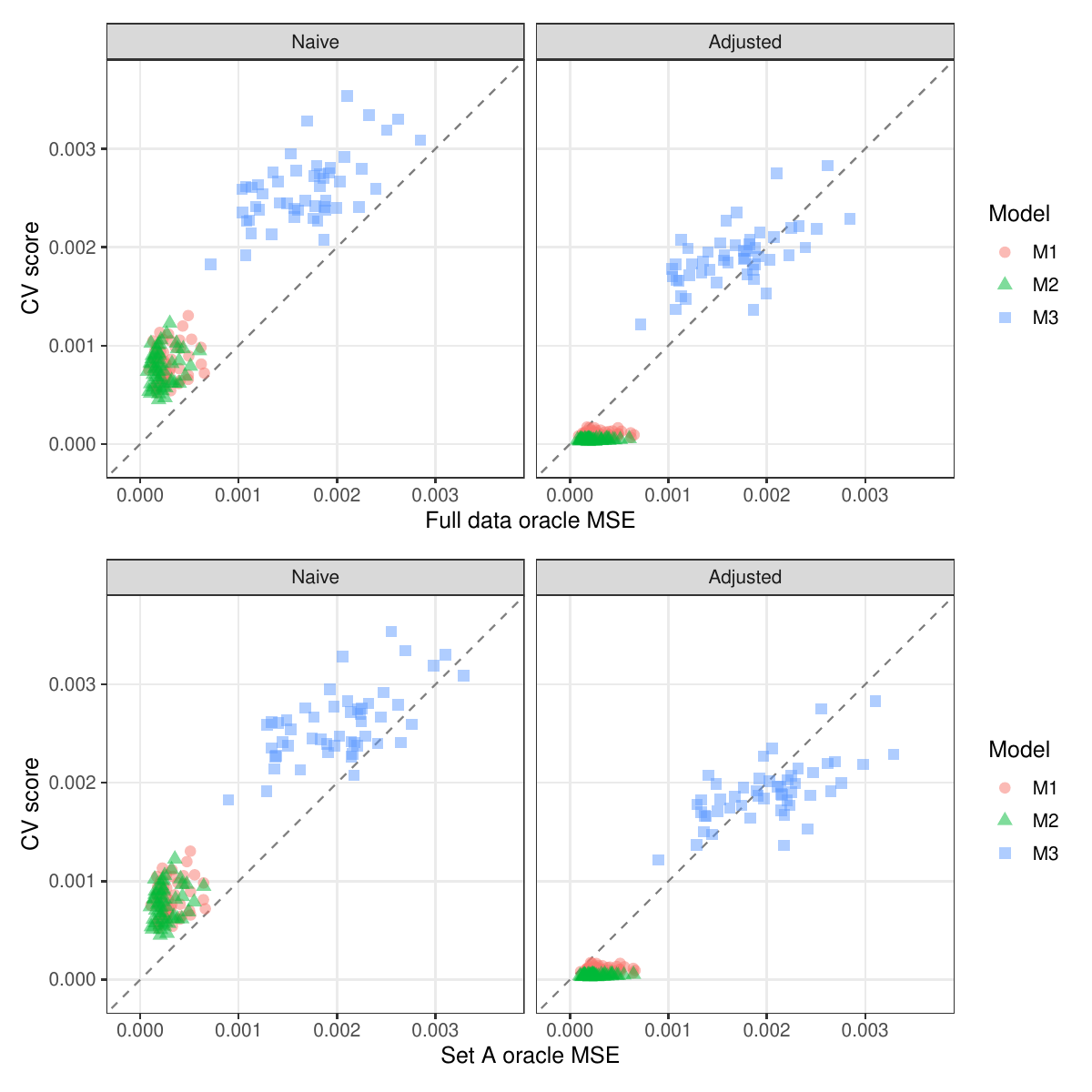}
\caption{Scatter plots of CV scores versus oracle errors under moderate sample size setting for the three models $\MM_1$, $\MM_2$, and $\MM_3$. Columns correspond to the naive and adjusted scores. The top row uses the full-sample oracle error, whereas the bottom row uses the Set A oracle error. Each point represents one simulation replicate, and the dashed $x=y$ line marks exact agreement between the CV score and the oracle error.}    \label{fig:S_score_scatter_national}
\end{figure}

\begin{figure}[!ht]
    \centering
    \includegraphics[width=.6\textwidth]{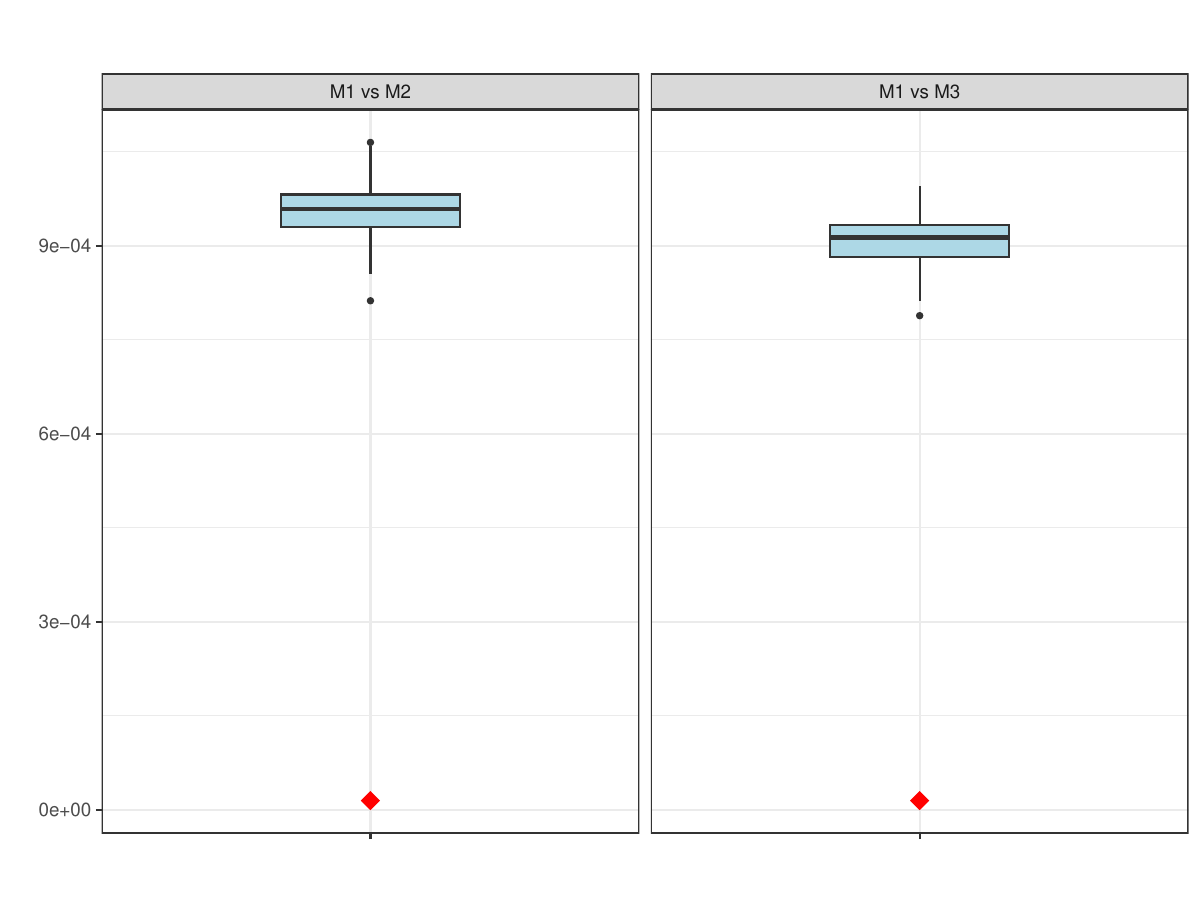}
    \caption{Distribution across the 50 simulation replicates of the plug-in error bound $\hat{t}_i$ (boxplots) and the absolute true remainder $|e_i|$ (red points), for the score-difference comparisons
$\mathcal{M}_1$ vs.~$\mathcal{M}_2$ (left) and $\mathcal{M}_1$ vs.~$\mathcal{M}_3$ (right) under the moderate sample-size setting.}
\label{fig:supp-bound-tightness}
    \label{fig:S_theo3}
\end{figure}

\begin{figure}[!ht]
    \centering
    \includegraphics[width=\textwidth]{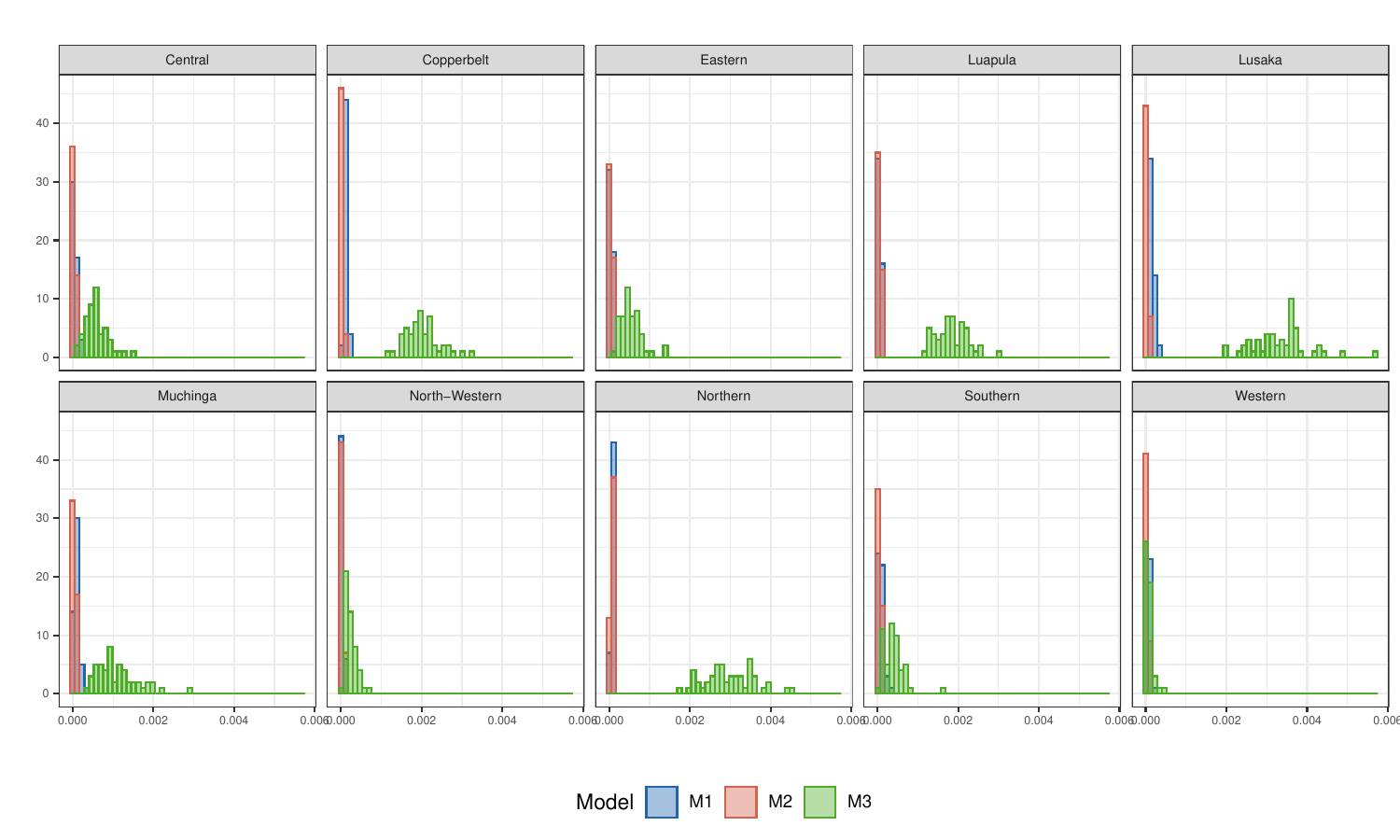}
 \caption{Area-level adjusted CV score distributions across ten provinces under moderate sample size setting, for model $\MM_1$, $\MM_2$, and $\MM_3$. Each histogram summarizes the distribution of area-specific adjusted CV scores over 50 
simulation replicates. The score distributions of $\MM_1$ and $\MM_2$ are nearly identical within each province, while $\MM_3$ is consistently shifted toward larger values.}
    \label{fig:S_score_hist_sub_h}
\end{figure}


\clearpage

\section{Model comparison using naive CV scores}

As an alternative to the proposed adjusted score comparison in the main paper, one may also evaluate the differences in naive CV scores of different models, and absorb the $v_i$ and $2c_i^\MM$ terms in the remainder bias term. We examine this approach in this section.

Similar to Theorem~\ref{thm3}, 
we can derive the bias for the naive score difference, as detailed in Theorem~\ref{thm:naive_diff} below, and a bound on the bias is provided in Theorem~\ref{thm:naive_bound}.

\begin{theorem}[Naive CV score difference]
\label{thm:naive_diff}
For the $i$-th area, the expected difference of naive CV scores satisfies
\begin{equation}\label{eq:naive_diff}
  \mathbb{E}\!\left[
    \widehat{\mathrm{score}}^{\mathcal{M}_1}_{\mathrm{naive},i}
    - \widehat{\mathrm{score}}^{\mathcal{M}_2}_{\mathrm{naive},i}
  \right]
  =
  \mathbb{E}\!\left[
    (\hat\theta^{\mathcal{M}_1}_{A,i} - \theta_i)^2
    - (\hat\theta^{\mathcal{M}_2}_{A,i} - \theta_i)^2
  \right]
  + f_i,
\end{equation}
where the expectation is over the sampling distribution of $S$ and the 
remainder is
\begin{equation}\label{eq:fi}
  f_i
  = -2\,\mathbb{E}\!\left[
      (\hat\theta^{w}_{B,i} - \theta_i)
      \bigl(\hat\theta^{\mathcal{M}_1}_{A,i} - \hat\theta^{\mathcal{M}_2}_{A,i}\bigr)
    \right].
\end{equation}
\end{theorem}

\begin{proof}
From Theorem~\ref{thm1}, taking the difference between $\mathcal{M}_1$ and $\mathcal{M}_2$, the 
model-free terms $\mathbb{E}[v_i]$ and $\mathbb{E}[b_i^2]$ cancel exactly, 
leaving
\begin{equation*}
  \mathbb{E}\!\left[
    \widehat{\mathrm{score}}^{\mathcal{M}_1}_{\mathrm{naive},i}
    - \widehat{\mathrm{score}}^{\mathcal{M}_2}_{\mathrm{naive},i}
  \right]
  =
  \mathbb{E}\!\left[
    (\hat\theta^{\mathcal{M}_1}_{A,i} - \theta_i)^2
    - (\hat\theta^{\mathcal{M}_2}_{A,i} - \theta_i)^2
  \right]
  - 2\mathbb{E}\!\left[c^{\mathcal{M}_1}_i - c^{\mathcal{M}_2}_i\right]
  - 2\mathbb{E}\!\left[b_i(d^{\mathcal{M}_1}_i - d^{\mathcal{M}_2}_i)\right].
\end{equation*}
Collecting the last two terms, and using the identity
$c^{\mathcal{M}}_i + b_i d^{\mathcal{M}}_i 
= \mathbb{E}[(\hat\theta^{\mathcal{M}}_{A,i} - \theta_i)
  (\hat\theta^w_{B,i}-\theta_i) \mid S]$
from Theorem~\ref{thm1}, the remainder becomes
\begin{equation*}
  f_i
  = -2\,\mathbb{E}\!\left[
      \mathbb{E}\!\left[
        (\hat\theta^{w}_{B,i} - \theta_i)
        \bigl(\hat\theta^{\mathcal{M}_1}_{A,i} - \hat\theta^{\mathcal{M}_2}_{A,i}\bigr)
        \,\Big|\, S
      \right]
    \right],
\end{equation*}
which equals $-2\,\mathbb{E}[(\hat\theta^w_{B,i}-\theta_i)
(\hat\theta^{\mathcal{M}_1}_{A,i}-\hat\theta^{\mathcal{M}_2}_{A,i})]$ 
by the law of iterated expectations.
\end{proof}

\begin{theorem}[Error bound for naive CV score difference]
\label{thm:naive_bound}
The remainder $f_i$ in Equation~\eqref{eq:fi} satisfies
\begin{equation}\label{eq:naive_bound_ineq}
  |f_i|
  \leq
  2\sqrt{
    \mathbb{E}\!\left[
      (\hat\theta^{w}_{B,i} - \theta_i)^2
    \right]
    \cdot
    \mathbb{E}\!\left[
      \bigl(\hat\theta^{\mathcal{M}_1}_{A,i} 
            - \hat\theta^{\mathcal{M}_2}_{A,i}\bigr)^2
    \right]
  }.
\end{equation}
For $K$-fold CV, a plug-in estimator of this bound is
\begin{equation}\label{eq:thatnaive}
 \hat{t}_{\text{naive},i}
  =
  2\sqrt{
    \left[
      \frac{1}{K}\sum_{k=1}^{K}
      \widehat{\mathrm{var}}\!\left(\hat\theta^{w}_{B_k,i}\right)
    \right]
    \left[
      \frac{1}{K}\sum_{k=1}^{K}
      \!\left(
        \hat\theta^{\mathcal{M}_1}_{A_k,i}
        - \hat\theta^{\mathcal{M}_2}_{A_k,i}
      \right)^{\!2}
    \right]
  }.
\end{equation}
\end{theorem}

\begin{proof}
Apply the Cauchy-Schwarz inequality to the conditional expectation given $S$:
\begin{equation*}
  \left|
    \mathbb{E}\!\left[
      (\hat\theta^w_{B,i}-\theta_i)
      \bigl(\hat\theta^{\mathcal{M}_1}_{A,i} - \hat\theta^{\mathcal{M}_2}_{A,i}\bigr)
      \,\Big|\, S
    \right]
  \right|
  \leq
  \sqrt{
    \mathbb{E}\!\left[(\hat\theta^w_{B,i}-\theta_i)^2 \,\Big|\, S\right]
    \cdot
    \mathbb{E}\!\left[
      \bigl(\hat\theta^{\mathcal{M}_1}_{A,i} 
            - \hat\theta^{\mathcal{M}_2}_{A,i}\bigr)^2
      \,\Big|\, S
    \right]
  }.
\end{equation*}
Take the outer expectation over $S$ and apply the Jensen's inequality to that expectation:

\begin{equation*}
  \mathbb{E}\!\left[
    \sqrt{
      \mathbb{E}[(\hat\theta^w_{B,i}-\theta_i)^2\mid S]
      \cdot
      \mathbb{E}[(\hat\theta^{\mathcal{M}_1}_{A,i}-\hat\theta^{\mathcal{M}_2}_{A,i})^2\mid S]
    }
  \right]
  \leq
  \sqrt{
    \mathbb{E}\!\left[\mathbb{E}[(\hat\theta^w_{B,i}-\theta_i)^2\mid S]\right]
    \cdot
    \mathbb{E}\!\left[
      \mathbb{E}[(\hat\theta^{\mathcal{M}_1}_{A,i}-\hat\theta^{\mathcal{M}_2}_{A,i})^2\mid S]
    \right]
  }
\end{equation*}
\end{proof}

For the plug-in estimator, the first factor in Equation~\eqref{eq:naive_bound_ineq} is $\mathbb{E}[(\hat\theta^w_{B,i} 
- \theta_i)^2] = \mathbb{E}[v_i] + \mathbb{E}[b_i^2]$. Since the H{\'a}jek estimator is approximately design-unbiased, $\mathbb{E}[b_i^2]$ is negligible relative to $\mathbb{E}[v_i]$, and 
$K^{-1}\sum_{k=1}^K\widehat{\mathrm{var}}(\hat\theta^w_{B_k,i})$ is an estimator of $\mathbb{E}[v_i]$.
The second factor $\mathbb{E}[(\hat\theta^{\mathcal{M}_1}_{A,i} - \hat\theta^{\mathcal{M}_2}_{A,i})^2]$ is the expected 
squared model difference. The fold average $K^{-1}\sum_k(\hat\theta^{\mathcal{M}_1}_{A_k,i} 
- \hat\theta^{\mathcal{M}_2}_{A_k,i})^2$ estimates this directly. Substituting both gives $\hat{t}_{\text{naive},i}$ in Equation~\eqref{eq:thatnaive}. The aggregation steps are the same as it for adjusted scores. Then in the full procedure, we compare the difference in naive scores to the error bound. If it exceeds the bound, we report the comparison based on the naive score.

\begin{figure}[!ht]
    \centering
    \includegraphics[width=.8\textwidth]{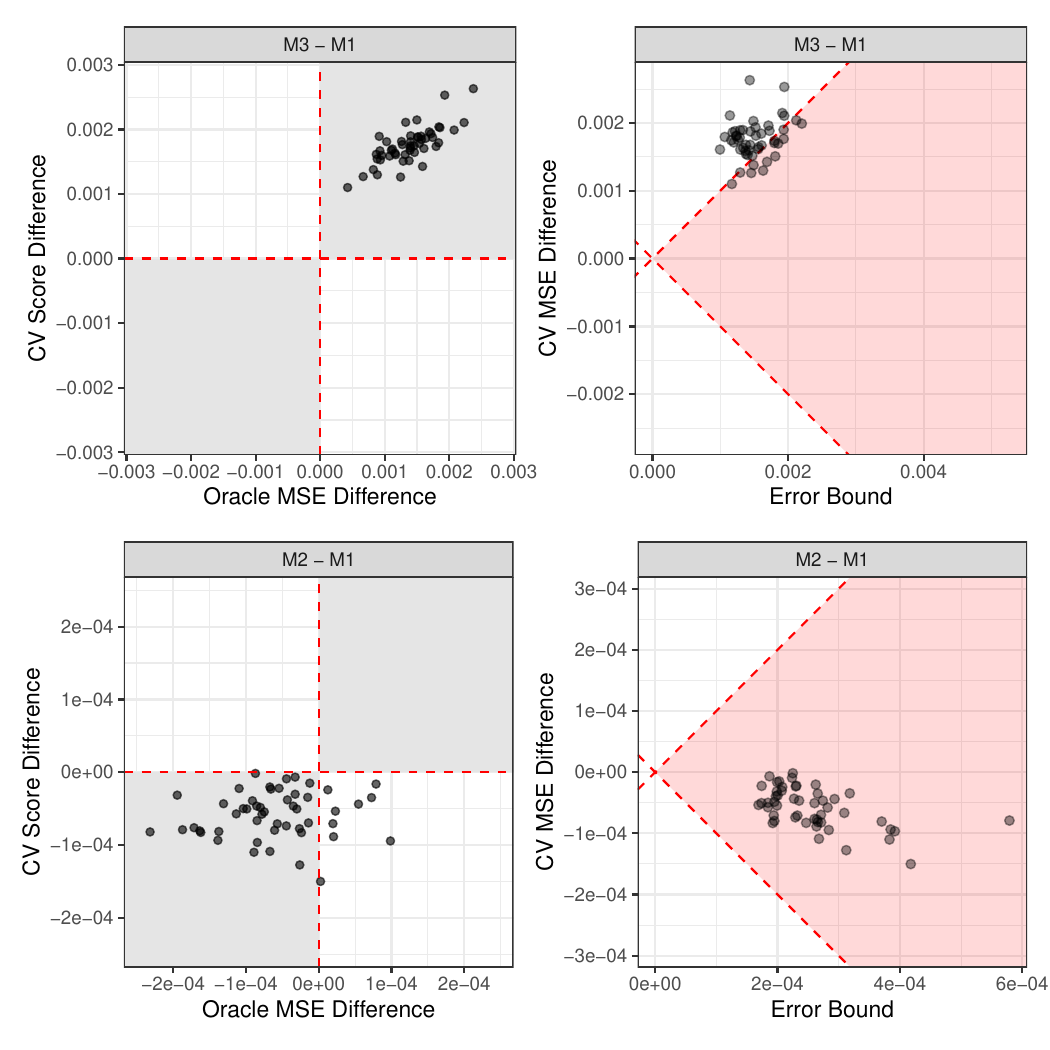}
 \caption{Results for province-level comparison under CV-SSU. \textit{Left:} Naive CV score differences versus oracle full-sample MSE differences. The $x$-axis shows the oracle full-sample MSE difference for $\mathcal{M}_3 - \mathcal{M}_1$ (top) and $\mathcal{M}_2 - \mathcal{M}_1$ (bottom). Each point represents one synthetic survey replicate. Points in the first and third quadrants indicate that the CV score difference has the same sign as the oracle full-sample MSE difference and thus correct ranking. \textit{Right:} Naive CV score differences versus the error bound for naive CV scores. The red shaded region is the inconclusive region where $ |\widehat{\mbox{score}}_q^{\MM a}-\widehat{\mbox{score}}_q^{\MM b}|< t_q$, i.e., score differences are too small to support decisive order of the two models.}
    \label{fig:S_LITbd51}
\end{figure}


Comparing to the error bound on adjusted score, $\hat{t}_{\text{naive},i}$ is more conservative. Figure~\ref{fig:S_LITbd51} illustrates the consequences. Here we used the moderate sample size setting as described in Section~\ref{sec:sim}. For the comparison between $\MM_3$ and $\MM_1$ , the naive score differences are positive in nearly all replicates under naive scores, correctly signing the oracle difference, but the bound $\hat{t}_{\text{naive},i}$ is relatively large to place some of replicates in the inconclusive region. For the comparison between $\mathcal{M}_2$ and $\mathcal{M}_1$, the model differences are again very small in scale and within the inconclusive region of naive score error bound.

\clearpage
\section{Cross-validation by PSU}
\label{sec:supp_cluster}

In this section, we consider an alternative sample splitting strategy for multi-stage stratified sampling design. Instead of randomly assigning SSUs within each PSU into $K$ folds, we randomly assign PSUs into $K$ folds. We refer to this strategy as CV-PSU, and the version presented in the main paper as CV-SSU.   

In principle, both CV strategies maintain the structure of the survey design after sample splitting. Figure \ref{fig:S_LIT-3050bd_diff_nntl_cluster} presents results for CV-PSU under the setting of moderate sample size (50 clusters per stratum, and 30 households per cluster). CV-PSU leads to the same conclusion as CV-SSU: $\MM_1$ and $\MM_2$ cannot be ranked conclusively, whereas the score difference between $\MM_1$ and $\MM_3$ is sufficiently large to conclude that $\MM_1$ outperforms $\MM_3$.


\begin{figure}[!ht]
    \centering
    \includegraphics[width=.8\textwidth]{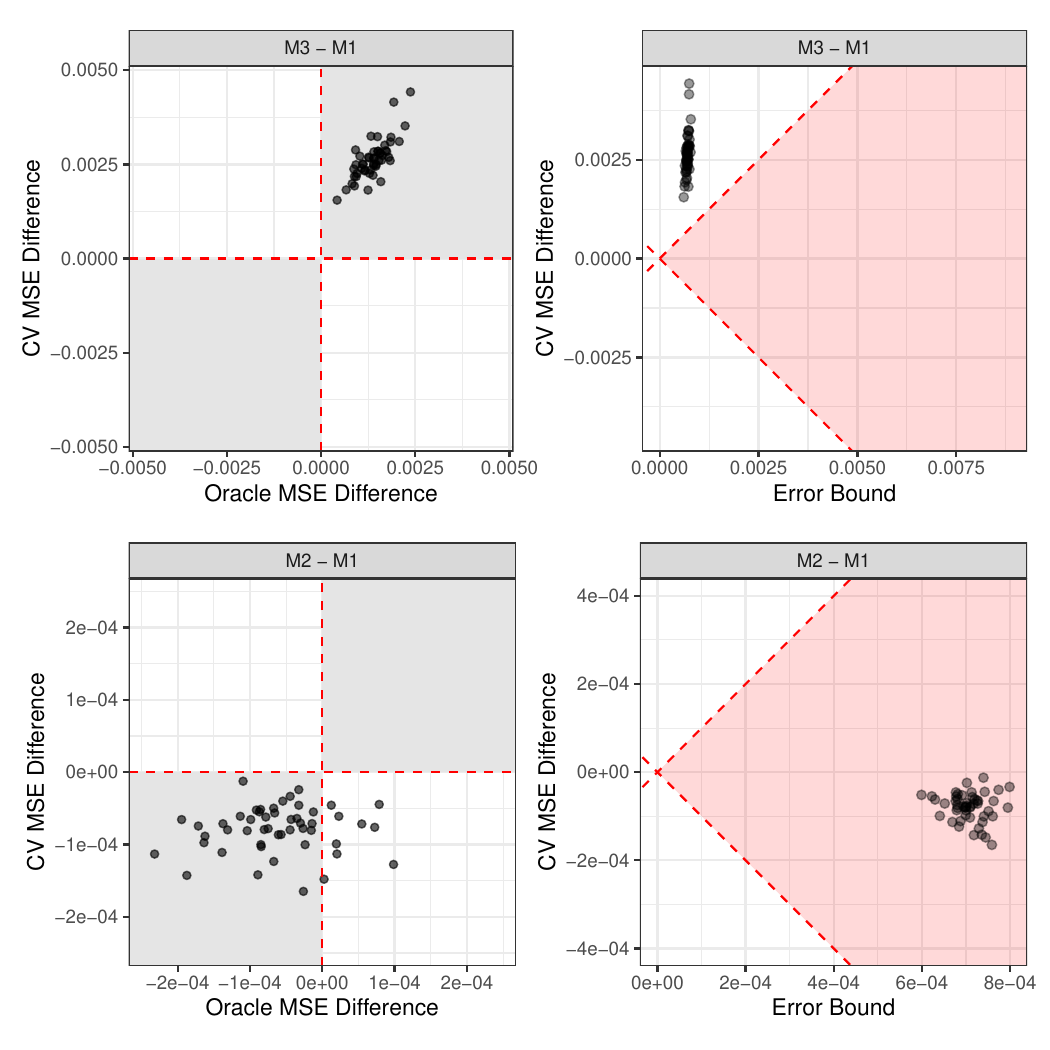}
 \caption{ CV-PSU: Simulation results under 50 clusters per stratum, 30 households per cluster,
\textit{Left:} Adjusted CV score differences versus oracle MSE differences.
The $x$-axis shows the oracle MSE difference $\mathcal{M}_3 - \mathcal{M}_1$ (top row) and $\mathcal{M}_2 - \mathcal{M}_1$ (bottom row); the $y$-axis shows the corresponding adjusted CV score differences. Each point represents one synthetic survey replicate. Points in the first and third quadrants indicate that the CV score difference has the same sign as the oracle MSE difference, implying correct ranking of the two models. \textit{Right:} Adjusted CV score differences versus the error bound. The red shaded region, bounded by the lines $y = x$ and $y = -x$, is the inconclusive zone where $ |\widehat{\mbox{score}}_q^{\MM a}-\widehat{\mbox{score}}_q^{\MM b}|< t_q$; score differences are too small to support decisive order of the two models.}
    \label{fig:S_LIT-3050bd_diff_nntl_cluster}
\end{figure}

Figure~\ref{fig:S_LIT-3050OR_comparison} highlights a key difference between CV-PSU and CV-SSU. Under $5$-fold cross-validation, removing one fifth of the clusters from a model causes greater information loss than removing one fifth of households within
each cluster. Thus for CV-PSU, the gap between the oracle full-sample MSE and oracle training MSE increases faster than CV-SSU as sample size decreases.
This can also be seen from Figure~\ref{fig:S_LIT-3050score_scatter_national_cluster} and \ref{fig:S_LIT-3040score_scatter_national_cluster}, where it shows that while the adjusted CV scores align relatively well with the training set oracle MSE, they could deviate from the oracle full-sample MSE when sample size is small. Therefore, in addition to the effect of the unidentifiable remainder term in estimating the training set oracle MSE, the performance of CV scores can also be impacted significantly by the gap between the two oracle MSE estimands. In contrast, the impact of reduced sample size on such gap is smaller for CV-SSU, as shown in Figure~\ref{fig:S_LIT-3040score_scatter_national}.

In summary, the oracle training MSE is a substantially poorer proxy for the full-sample oracle under CV-PSU than under CV-SSU, which can lead to more bias in model ranking.


\begin{figure}[!ht]
    \centering
    \begin{subfigure}{0.8\textwidth}
        \centering
        \includegraphics[width=\textwidth]{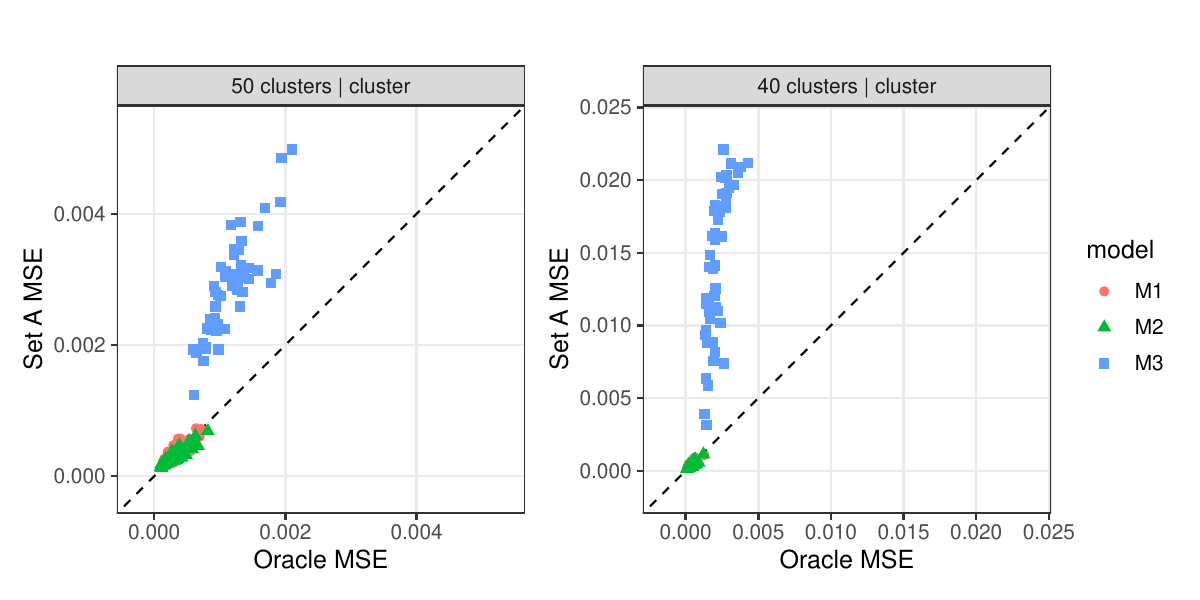}
        \caption{CV-PSU}
    \end{subfigure}
    
    \vspace{0.5em}
    
    \begin{subfigure}{0.8\textwidth}
        \centering
        \includegraphics[width=\textwidth]{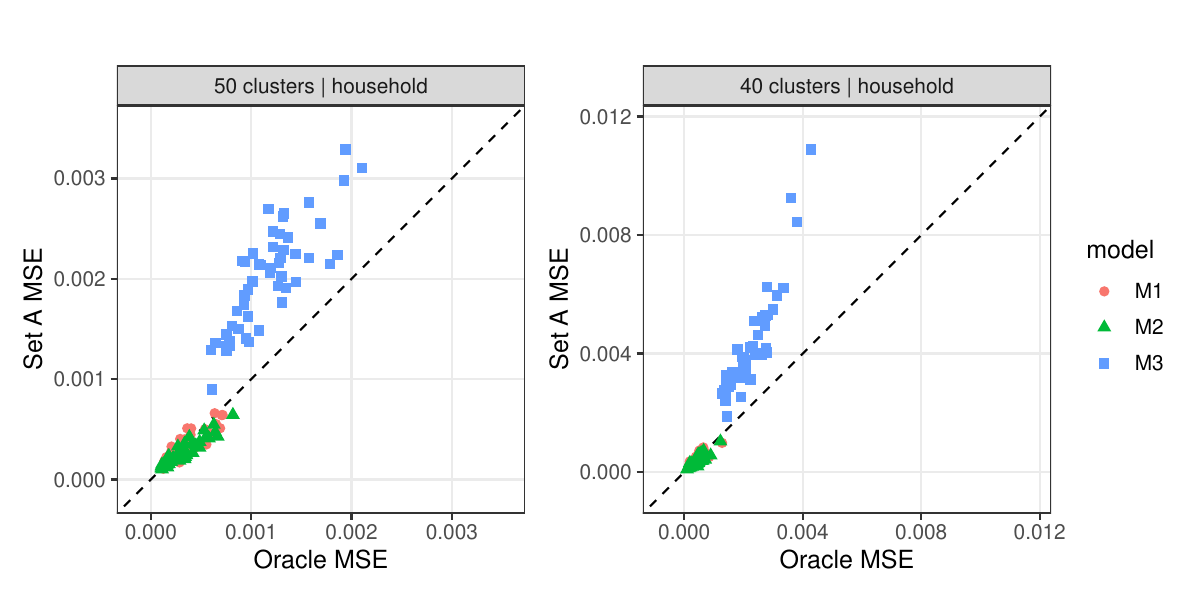}
        \caption{CV-SSU}
    \end{subfigure}
    
    \caption{Comparison of full-sample oracle MSE and oracle training MSE under two cross-validation schemes. Left: 50 clusters per stratum, each with 30 households. Right: 40 clusters per stratum, each with 30 households. Top row: CV-PSU. Bottom row: CV-SSU.}
    \label{fig:S_LIT-3050OR_comparison}
\end{figure}

\begin{figure}[!ht]
    \centering
    \includegraphics[width=.8\textwidth]{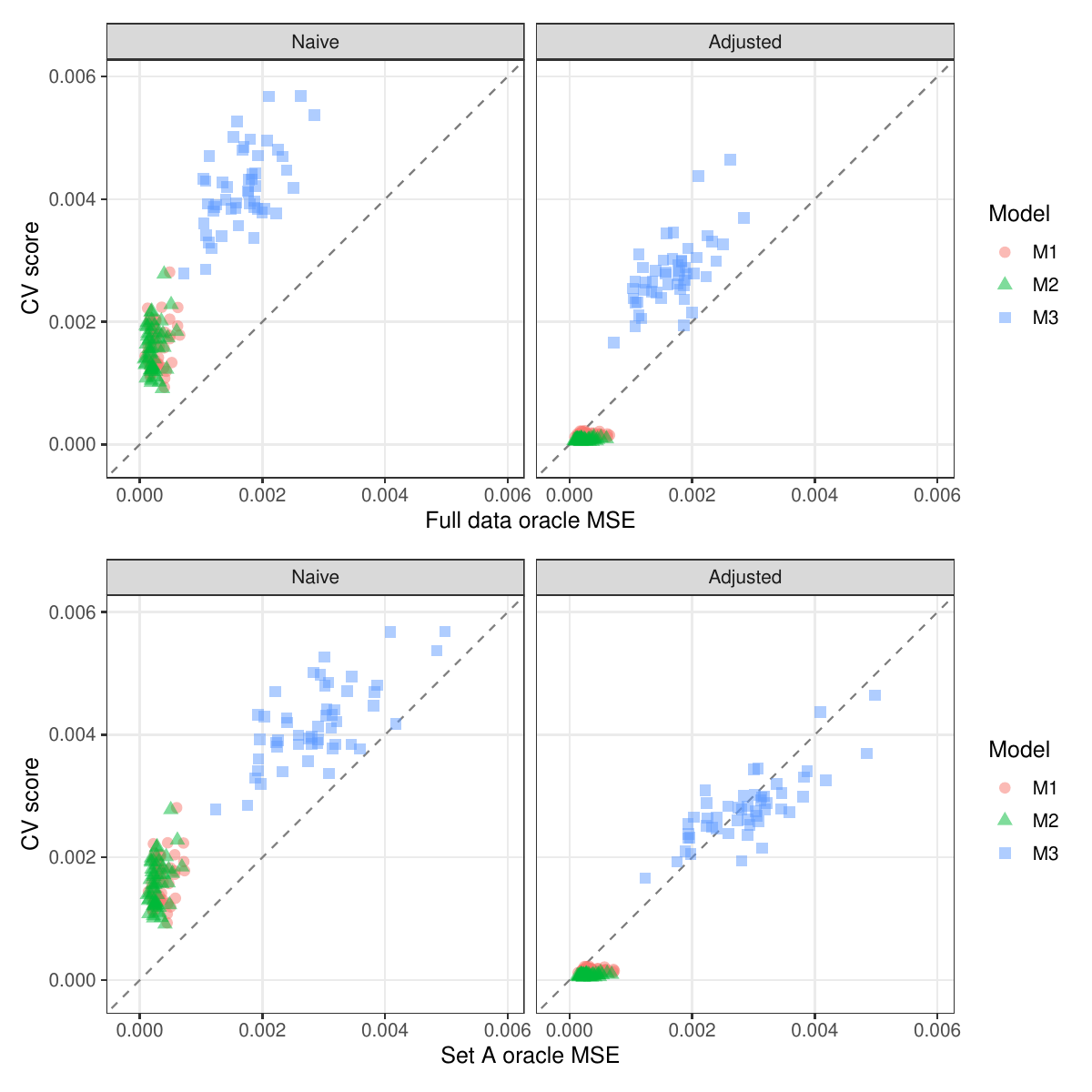}
 \caption{Naive (left) and adjusted (right) CV scores against the oracle MSE under CV-PSU, with 50 clusters per stratum and 30 households per cluster. The top row uses the oracle full-sample MSE as reference; the bottom row uses the oracle training MSE. Each point represents one simulation replicate.}
    \label{fig:S_LIT-3050score_scatter_national_cluster}
\end{figure}

\begin{figure}[!ht]
    \centering
    \includegraphics[width=.8\textwidth]{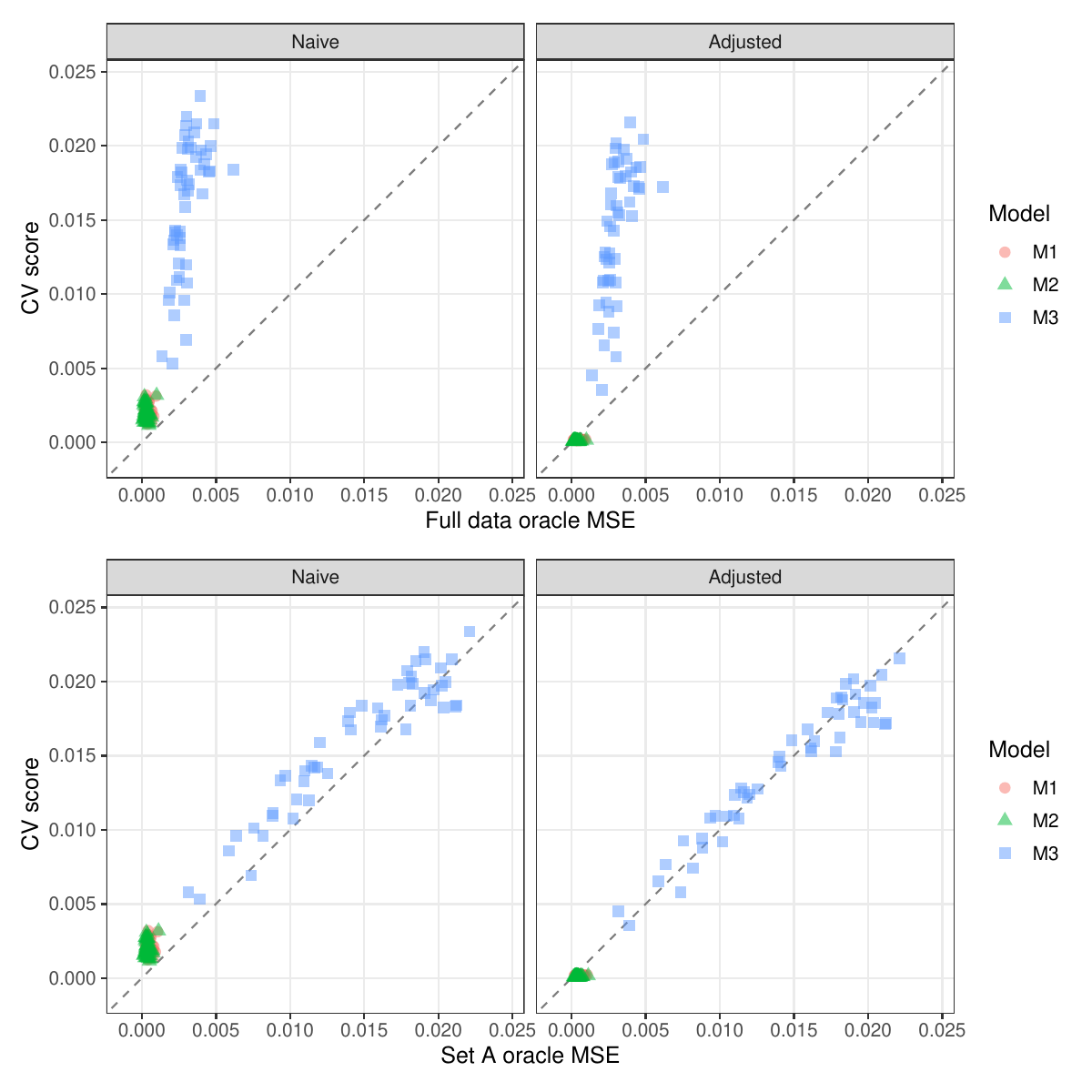}
 \caption{Naive (left) and adjusted (right) CV scores against oracle MSE under CV-PSU, with 40 clusters and 30 households per cluster. The top row uses the oracle full-sample MSE as reference; the bottom row uses the oracle training MSE. Each point represents one simulation replicate.}
    \label{fig:S_LIT-3040score_scatter_national_cluster}
\end{figure}

\begin{figure}[!ht]
    \centering
    \includegraphics[width=.8\textwidth]{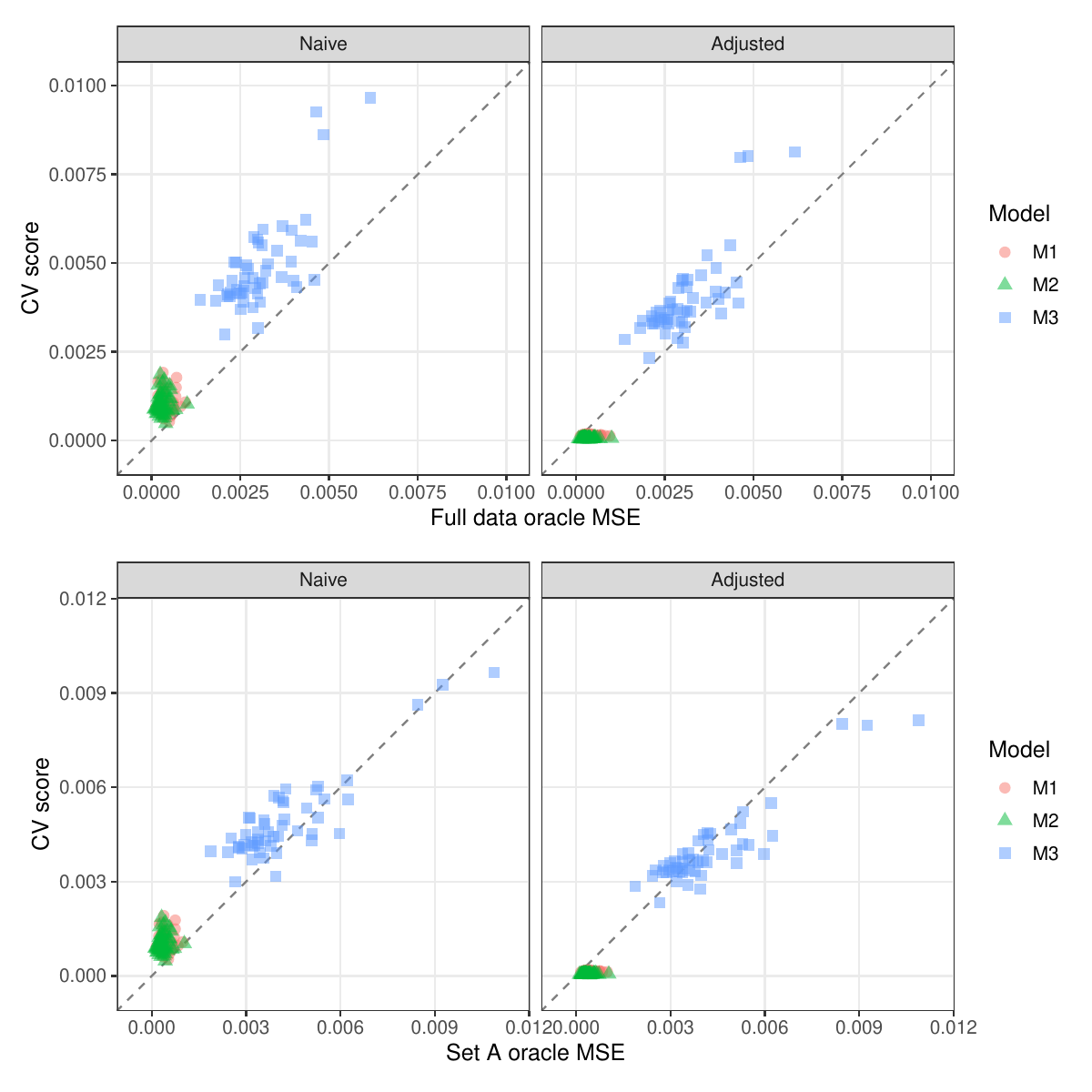}
 \caption{Naive (left) and adjusted (right) CV scores against oracle MSE under CV-SSU, with 40 clusters and 30 households per cluster. The top row uses the oracle full-sample MSE as reference; the bottom row uses the oracle training MSE. Each point represents one simulation replicate.}
    \label{fig:S_LIT-3040score_scatter_national}
\end{figure}

%

\clearpage

\section{Two-fold CV-SSU with resplitting}
\label{sec:supp_2f}

In this section, we discuss an alternative to $5$-fold CV presented in the main paper. Under two-fold CV, the sample $S$ is randomly partitioned into two equal halves, with one serving as the training set $A$ and the other as the validation set $B$. The partition can be repeated multiple times to reduce sensitivity to the choice
of split. Because each training set retains only half the data, the dependence between $\hat\theta^{\mathcal{M}}_{A,i}$ and $\hat\theta^w_{B,i}$ is weaker than under $K$-fold CV, where the training set retains $(K-1)/K$ of the sample.
Consequently, the covariance term $\hat{c}^{\mathcal{M}}_i$ in Theorem~\ref{thm2} is smaller in magnitude, and the naive and adjusted scores are closer to each other than under 5-fold CV. The terms $\hat{v}_i$, $\hat{c}^{\mathcal{M}}_i$, and the naive squared errors are averaged over the repeated pairs before computing the adjusted score.
\begin{figure}[!ht]
    \centering
    \begin{subfigure}[t]{\textwidth}
        \centering
        \includegraphics[width=1\textwidth]{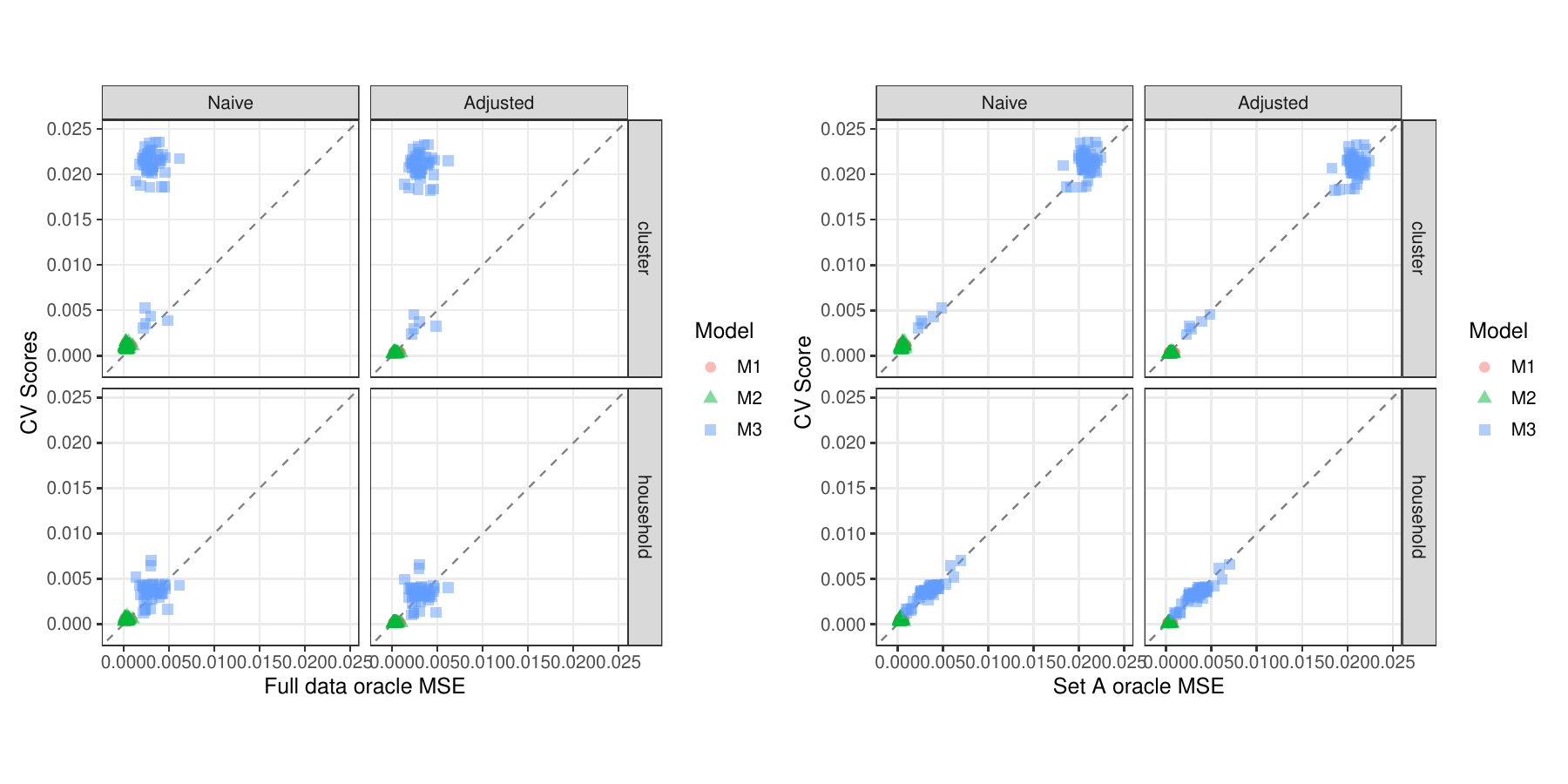}
        \caption{Moderate sample size: 50 clusters per stratum, 30 households per cluster.}
        \label{fig:S_LIT-3050score_scatter_national_2f}
    \end{subfigure}
    \vspace{1em}
    \begin{subfigure}[t]{\textwidth}
        \centering
        \includegraphics[width=1\textwidth]{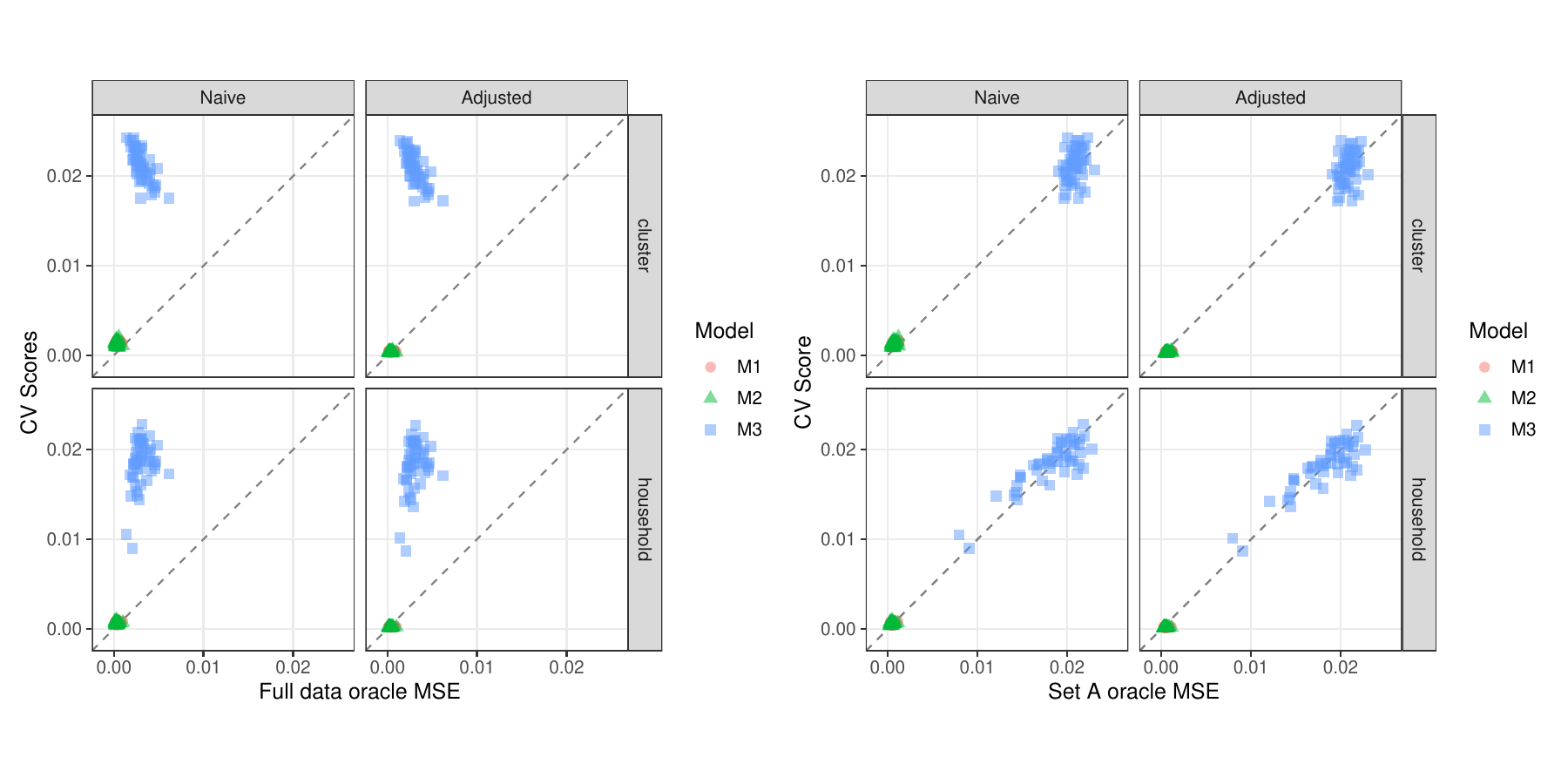}
        \caption{Small sample size: 40 clusters per stratum, 30 households per cluster.}
        \label{fig:S_LIT-3040score_scatter_national_2f}
    \end{subfigure}
    \caption{Scatter plots of 2-fold CV scores versus oracle MSE for the three models
$\MM_1$, $\MM_2$, and $\MM_3$, under CV-PSU (top row of each panel) and CV-SSU
(bottom row of each panel). Columns correspond to the naive and adjusted scores. The left
half of each panel uses the full-sample oracle MSE as the reference; the right half uses
the training oracle MSE. Each point represents one simulation replicate, and the dashed
$x = y$ line marks exact agreement between the CV score and the oracle MSE.}
    \label{fig:S_2f_combined}
\end{figure}

In this simulation, we repeat the partition five times and compare both CV-PSU and CV-SSU under the moderate-sample setting (50 clusters per stratum, 30 households per cluster) and the small-sample setting (40 clusters per stratum, 30
households per cluster). Under all settings, two-fold CV-SSU with resplitting correctly ranks the three models: $\MM_3$ is the least favored and $\MM_1$ and $\MM_2$ are very close. Additionally, both CV-PSU and CV-SSU estimate the training
oracle error accurately. 

The differences again lie in how closely the training oracle MSE tracks the full-data oracle in these scenarios. Under the moderate-sample setting, CV-SSU tracks the full-sample oracle more closely than CV-PSU, consistent with the
findings in Supplementary Section~\ref{sec:supp_cluster}. Under the small-sample setting (40 clusters), the oracle training MSE for $\MM_3$ departs from the full-sample oracle under both strategies, but the divergence is more severe for CV-PSU.
Taken together, two-fold CV-SSU is more reliable than two-fold CV-PSU for tracking the full-sample MSE, but less so than 5-fold CV-SSU at small sample sizes, because each training fold retains only half the data and the
resulting oracle training MSE estimand is a noisier proxy for the full-sample target.

\FloatBarrier


\section{LOAO analysis of the 2024 Zambia DHS}

\begin{figure}[!ht]
    \centering
    \includegraphics[width=1\textwidth]{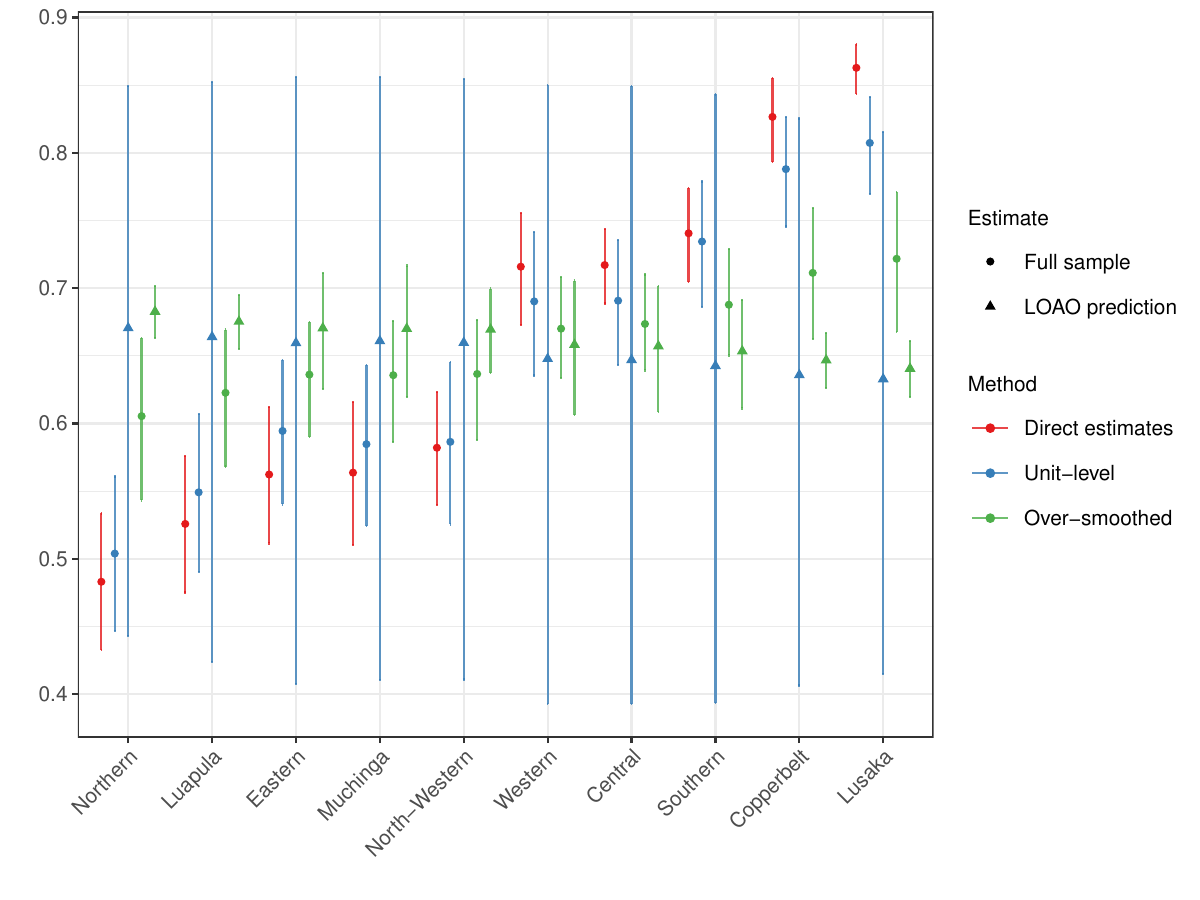}
  \caption{Leave-one-area-out validation at the province level under an IID random-effects specification. Points and 90\% credible intervals are shown for the unit-level model and the over-smoothed model under both full-data and LOAO fits, alongside design-based direct estimates with 90\% confidence intervals.}
  \label{fig:loao-admin1}
\end{figure}

\begin{figure}[!ht]
    \centering
    \includegraphics[width=1\textwidth]{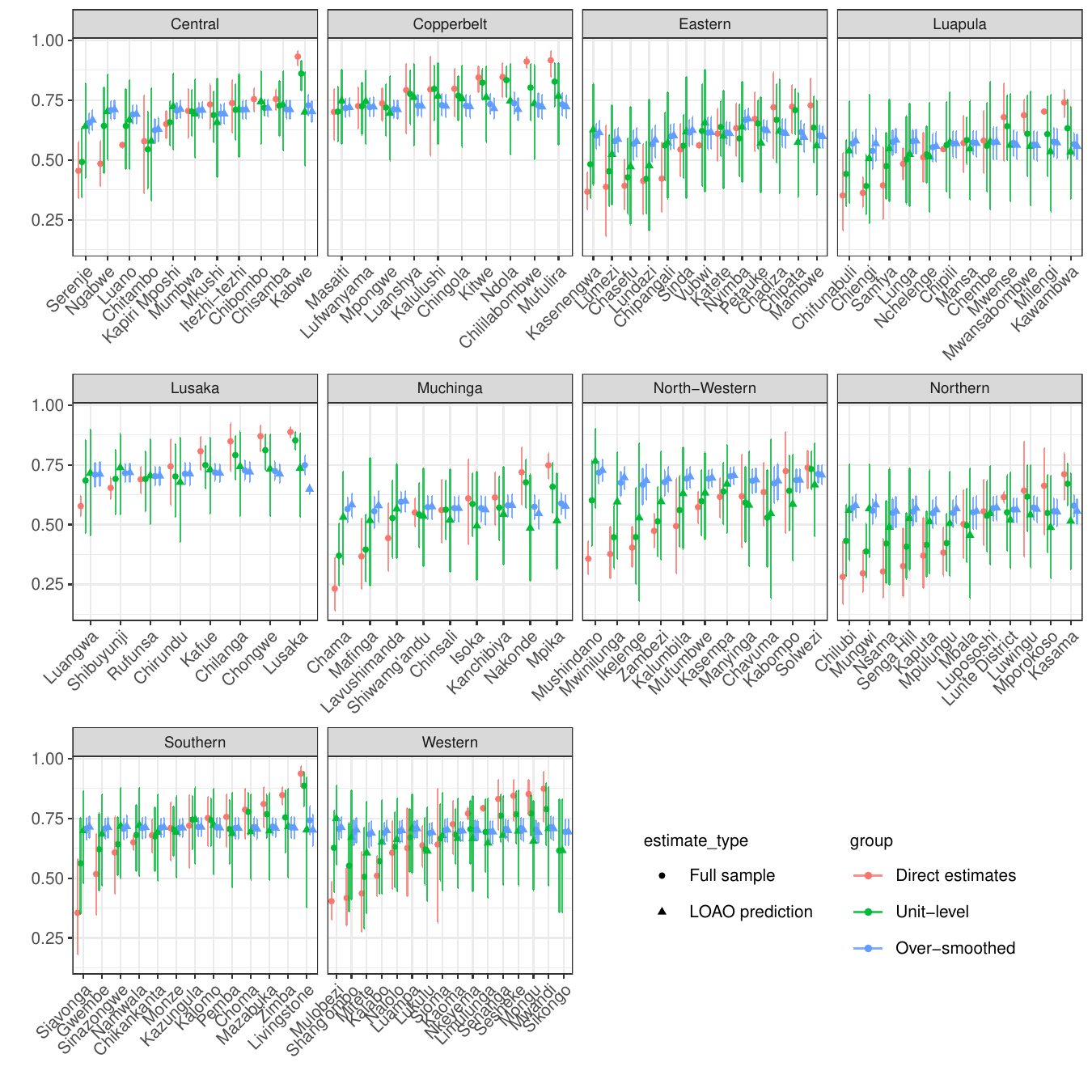}
  \caption{Leave-one-area-out validation at the Admin-2 level under a BYM2
spatial random-effects specification, with panels arranged by province. Points and 90\% credible intervals are shown for the unit-level model and the over-smoothed model under both full-data and LOAO fits, alongside the design-based direct estimates with 90\% confidence intervals.}
  \label{fig:loao-admin2}
\end{figure}

We apply LOAO to the 2024 Zambia DHS at two administrative levels for $\MM_1$ and $\MM_3$ described in the main paper. At the province level, the LOAO score is $0.0240$ for $\MM_3$ and $0.0233$ for $\MM_1$, and thus $\MM_1$ is slightly preferred by the LOAO cross-validation. At the district level, the LOAO score is $0.0179$ for $\MM_3$ and $0.0256$ for $\MM_1$, yielding the $\MM_3$ is preferred.

To gain more insight into the behavior of LOAO, Figures~\ref{fig:loao-admin1} and~\ref{fig:loao-admin2} compares the model estimates from the full sample and the LOAO predictions. The full-sample estimates $\hat\theta_i^{\MM}$ and the LOAO predicted estimates $\hat\theta_i^{(-i)}$ show dramatically different behavior at the province level with IID random effect in Figure~\ref{fig:loao-admin1}. This renders the MSE estimates based on the latter extremely biased. Such difference is less pronounced at the district level with spatially structured random effects in the model, as shown in Figure~\ref{fig:loao-admin2}. However, the correct model ranking at the district level is more driven by the severe over-smoothing of $\MM_3$. 

\section{Additional results for the analysis of the 2024 Zambia DHS}

Figure \ref{fig:zm_ad1ad2map_mapsdl} shows the standard deviation for the corresponding estimates in Figure \ref{fig:literacymap} of the main paper. Across both levels, the over-smoothed model $\MM3$ yields the smallest standard deviations, reflecting the underestimated uncertainty as an over-smoothed model. 

\begin{figure}[!ht]
    \centering
    \includegraphics[width=1\textwidth]{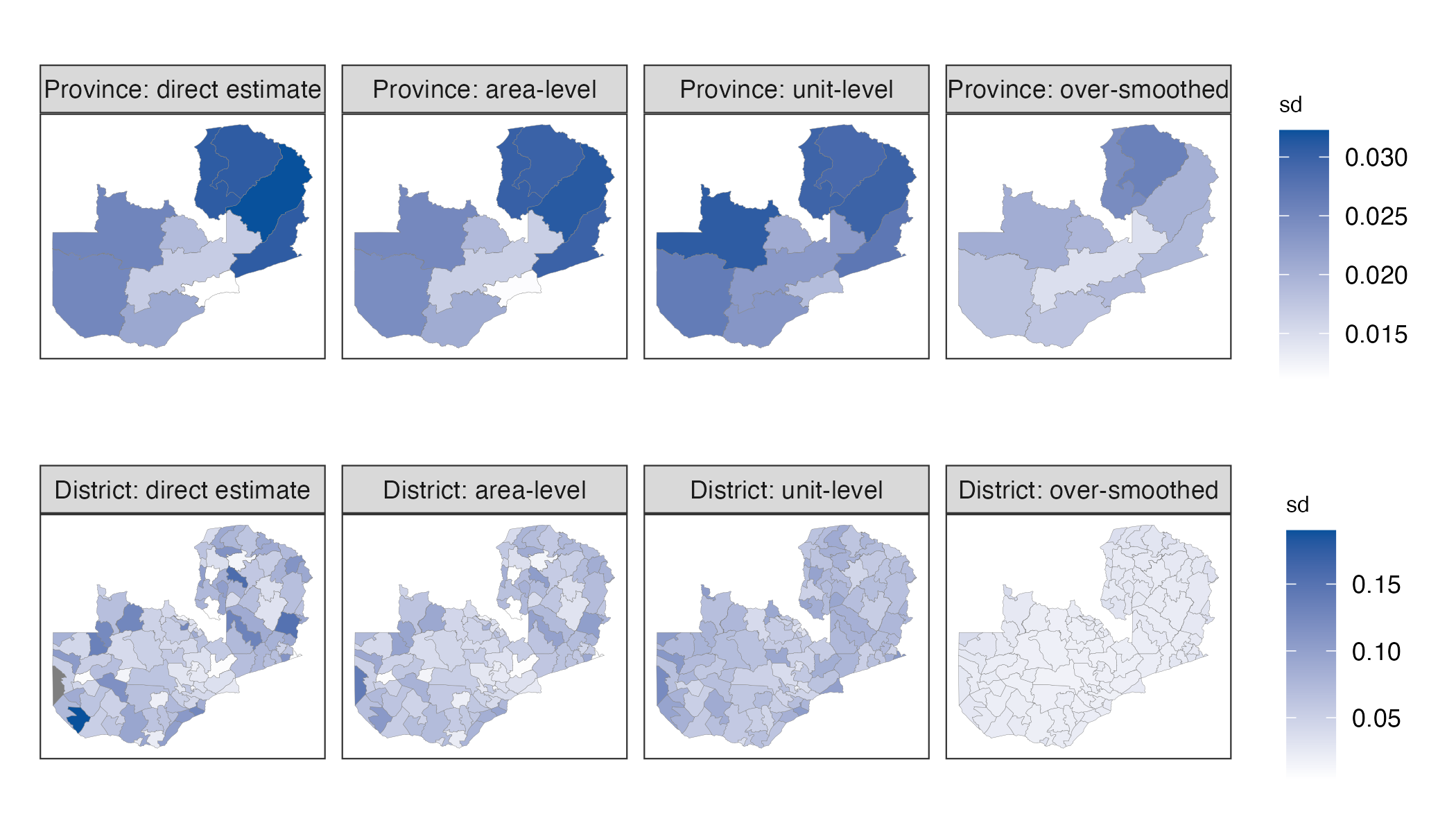}
    \caption{Province- and district-level standard deviation of the female literacy rate under direct estimation and the three candidate SAE models considered in Section \ref{sec:real-data}}
    \label{fig:zm_ad1ad2map_mapsdl}
\end{figure}

\end{document}